\documentclass[11pt,a4paper]{article}

\usepackage{url}
\usepackage{graphicx}
\usepackage{dcolumn}
\usepackage{bm}
\usepackage[all,cmtip]{xy}
\usepackage{comment}
\usepackage{color}
\usepackage[american]{isodate}
\usepackage{amsmath}
\usepackage{amsfonts}
\usepackage{amssymb}
\usepackage{amsthm}
\usepackage{authblk}
\usepackage{fullpage}

\renewcommand{\vec}[1]{{\bf #1}}
\providecommand{\keywords}[1]{\textbf{\textit{Index terms---}} #1}
\newcommand{\argmax}{\operatorname{argmax}}

\newcommand{\fig}[1]{Fig.~(\ref{fig:#1})}
\newcommand{\eq}[1]{Eq.~(\ref{eq:#1})}
\renewcommand{\a}{\alpha}
\newcommand{\g}{\gamma}

\renewcommand{\b}{\beta}

\newcommand{\sg}{\sigma}

\renewcommand{\k}{\kappa}
\renewcommand{\t}{\theta}

\newcommand{\NH}{\mathbb{N}}

\newcommand{\RH}{\mathbb{R}}

\usepackage{tikz-cd}

\newcommand{\set}[1]{\lbrace #1 \rbrace}

\newcommand{\ra}{\rightarrow}

\newcommand{\lgans}{\textquotedblleft}

\newcommand{\nn}{\nonumber}

\numberwithin{equation}{section}      
\theoremstyle{plain}
\newtheorem{thm}{Theorem}[section]  
\newtheorem{lem}[thm]{Lemma} 
 
\newtheorem{cor}[thm]{Corollary}

\theoremstyle{definition}

\begin{document}
\title{Inferring Volatility in the Heston Model and its Relatives -- an Information Theoretical Approach}

\author{Nils Bertschinger}
\affil{Frankfurt Institute for Advanced Studies, Routh-Mofang-Stra{\ss}e 1, 60483 Frankfurt am Main}
\author{Oliver Pfante}
\affil{Frankfurt Institute for Advanced Studies, Routh-Mofang-Stra{\ss}e 1, 60483 Frankfurt am Main, Germany. Email: pfante@fias.uni-frankfurt.de, Telephone: +49 69 798 475 29}

\date{\today}
\maketitle
\begin{abstract}
  Stochastic volatility models describe asset prices $S_t$ as driven
  by an unobserved process capturing the random dynamics of volatility
  $\sigma_t$. Here, we quantify how much information about $\sigma_t$
  can be inferred from asset prices $S_t$ in terms of Shannon's mutual
  information $I(S_t : \sigma_t)$.

  This motivates a careful numerical and analytical study of
  information theoretic properties of the Heston model.  In addition,
  we study a general class of discrete time models motivated from a
  machine learning perspective. In all cases, we find a large
  uncertainty in volatility estimates for quite fundamental
  information theoretic reasons.

\end{abstract}

\keywords{Information Theory; Stochastic Volatility; Bayesian Analysis}

\section{Introduction}

Many plaudits have been aptly used to describe Black and Scholes' \cite{Black1973} contribution to option pricing theory. However, especially after the 1987 crash, the geometric Brownian motion model and the Black-Scholes formula were unable to reproduce the option price data of real markets. This is not surprising, since the Black-Scholes model makes the strong assumption that log-returns of stocks are normally distributed with a volatility which is not only assumed to be known but also constant over time. Both assumptions on volatility are wrong: first, volatility is a hidden parameter which needs to be inferred from stock and option data, respectively; second, this so called ``implied volatility'' is not constant at all but a highly volatile time-process. The second insight led to the introduction of implied volatility indices like the VIX (1993) and its off-springs, based on the work \cite{Brenner1989,Brenner1993}, which make implied volatility a trademark in its own right subject to similar stochastic movements as stock prices. Among the most relevant statistical properties of these implied volatility processes, volatility seems to be responsible for the observed clustering in price changes. That is, large fluctuations are commonly followed by other large fluctuations and similarly for small changes \cite{Bouchaud2003}. Another feature is that, in clear contrast with price changes which show negligible autocorrelations, volatility autocorrelation is still significant for time lags longer than one year \cite{Perello2008,Muzy2000,Bouchaud2003,Lo1991,Ding1993,LeBaron2001}. Additionally, there exists the so-called leverage effect, i.e., much shorter (few weeks) negative cross-correlation between current price change and future volatility \cite{Bouchaud2001,Bouchaud2003,Black1973,Bollerslev2006}.\\
Accompanying these empirical findings on volatility processes there is also a long history of theoretical attempts to map the time dependency of volatility via stochastic processes -- see \cite{Shepard2005} for an extensive review of the literature. The first stochastic volatility model directly addressing the observed effect of volatility clustering is the one of Taylor \cite{Taylor1982} -- a univariate time-discrete autoregression. \cite{Nelson1994} investigate those ARCH models and their generalized descendants, GARCH and EGARCH, how they estimate conditional variances and covariances properly even in the case they are misspecified.\\
Continuous time models were introduced by Johnson \cite{Johnson1979}. The best known paper in this area is Hull and White \cite{HullWhite1987}. Heston's model \cite{Heston1993}, giving rise to a gamma distributed volatility, is very popular, thanks to the fact that it has a closed-form expression for the characteristic function of its transitional probability. Research in the late 1990s and early 2000s has shown that more complicated stochastic volatility models are needed to model either option or high frequency data since the previously mentioned models do not deal properly with significant changes of volatility on a short time scale. Thus, jumps were incorporated in stochastic volatility \cite{Shepard2001,Polson2003} and in \cite{D.DuffieJ.Pan2000} affine jump-diffusion processes are treated analytically deriving closed form solutions for call options as Heston obtained for his model \cite{Heston1993}. There are also approaches which model the volatility process as a function of a number of separate stochastic processes or factors. In \cite{Tauchen2003} two-factor models are studied with one very slowly mean reverberating factor and another quickly mean reverting one. Such two-factor models provide an alternative to jump-diffusions in breaking the link between tail thickness and volatility persistence. \cite{Masoliver2005} investigates an exponential Ornstein-Uhlenbeck stochastic volatility model and observes that the model shows a multi-scale behaviour in the volatility autocorrelation. It also exhibits a leverage correlation and a probability profile for the stationary volatility which are consistent with market observations. All these features make the model quite appealing since it appears to be more complete than other stochastic volatility models also based on a two-dimensional diffusion discussed in \cite{Tauchen2003}. In summary, there is a vast literature on stochastic volatility models investigating and discussing to which extend those models fit with the data and reproduce the statistical properties of realised volatility of stocks: temporal clustering, leverage effect, volatility autocorrelation, empirical option prices, and the temporal evolution of volatility. Furthermore, papers like \cite{Tauchen2003} and \cite{Nelson1994} strive identify the best model on empirical grounds. I.e., these papers use statistical tests to decide which stochastic volatility models do best in fitting real world data, such as stock prices, volatility indices, etc. Hence, the entire literature reviewed in \cite{Shepard2005} tackles mainly the second previously listed issue with volatility in the classical Black-Scholes model: its time dependency.  \\
The approach presented in this paper has a different focus: Assuming that a stochastic volatility model is correct, i.e., stock prices and daily returns, respectively, follow its dynamics, how reliably can
the hidden volatility inferred from those data according to the stochastic volatility model describing this
data precisely. That is, compared to the previously cited papers, we particularly stress the problem arising from the first mentioned problem concerning Black and Scholes model: the hidden nature of volatility and the only possibility to grasp it via implied volatility derived from market data. If stochastic volatility models describe financial markets properly, how reliable is implied volatility as a volatility measurement. That is, how much information about volatility can be derived from stock data as its observable? Since stochastic volatility models deal with two different distributions for the stock process $S_t$ and its volatility process $\sg_t$, Shannon's information theory \cite{Shannon1948} provides an ideal frame to make this question precise: What is the mutual information $I(S_t : \sg_t)$ between the observed stock prices $S_t$ and the hidden volatility process $\sg_t$? Among the vast amount of stochastic volatility models choices have to be made. First, we consider Heston's model \cite{Heston1993} because the closed form solution of the transitional probabilities of the joint stock and volatility process is analytically appealling. Second, the purely theoretical studies on the Heston model are accompanied by an empirical survey on exponential Ornstein-Uhlenbeck single- and two-factor model which are suggested by \cite{Masoliver2005} as stochastic volatility models fitting observed data best. Hence, the models have been chosen with respect to their claimed theoretical and empirical advantages, respectively, even though we believe the results obtained in the present paper to prevail among a wide class of models. \\
These results are quite disillusioning. Via numerical computations we were able to compute the evolution of the mutual information $I(\sg_t : S_t)$ over time in the Heston model for parameters derived from the VIX in \cite{Ait-Sahalia2007}. It is at most about $0.5$ bit and seems even to diminish for $t \ra \infty$.  
A further interesting finding is the existence of a local maximum of the function $t \ra I(\sg_t : S_t)$ suggesting that there is an intermediate time scale at which the two processes are most closely coupled.
Next, we study a different class of stochastic volatility models, namely the exponential Ornstein-Uhlenbeck model. Here, we generalize the class of models even further by considering arbitrary Gaussian processes driving the logarithm of the volatility. This allows us to utilize powerful tools from machine learning to fit such models to stock price data and compare their relative performance. While we can replicate some of the findings regarding stylized facts on volatility dynamics, e.g. long-range correlations, we again demonstrate that stock prices provide only limited information about the volatility. This obviously has severe implications for volatility predictions. We illustrate our results using different data sets and relate them to the computations for Heston's model. \\
Our paper is structured as follows: First, we present the Heston model and related stochastic volatility models. The third section presents a short overview of information theory and its crucial ingredients: differential entropy, mutual information and its scaling invariance. 
Additionally, we present a multilevel dynamical systems point of view on Heston's model. If one considers the joint process $t \ra (S_t, \sg_t)$ as a microscopic process, and the stock process $t \ra S_t$ as its observable, the previous question whether the inference from stock $S_t$ to volatility $\sg_t$ is reliable leads to a closely related question: to which extend is the stock $t \ra S_t$ a process in its own right? That is, how much does the Heston model differ from the classical Black Scholes model which assumes the Stock process $t \ra S_t$ to be a Markovian process in its own right. Information measures developed by the authors and others in \cite{Pfante2014} address this question. 
The fourth section presents the numerics on the S{\&}P 500 with parameters from \cite{Ait-Sahalia2007} as well as fits of exponential stochastic volatility models. We then conclude by discussing our findings in the light of volatility forecasting.

\section{Stochastic Volatility Models}

There are many versions of Heston's \cite{Heston1993} model depending on whether they care about volatility premium of the risk-neutral model or not. It is an ongoing debate whether it is significant. Evidence is provided by \cite{Hurn2012} whereas in \cite{Ait-Sahalia2007} it is considered statistically insignificant. Whatever the truth might be, the debate demonstrates that the quantitative effects of volatility premium are small. Since our analysis is devoted to the qualitative aspects of Heston's model we adopt its reduced form and follow mainly \cite{Dragulescu2002}.\\
As in most stochastic volatility models, we consider a  stock, whose price $t \ra S_t$ process follows a geometric Brownian motion:
\[
dS_t= \mu S_t dt + \sg_t S_t dW_t^{(1)} \, 
\]
with a fixed initial value $S_0$. $\mu$ is the drift parameter, $W^{(1)}_t$ a standard Brownian motion, $\sg_t$ the time dependent volatility. Since the geometric Brownian motion only depends on $\sg^2_t$ we introduce the variance $v_t = \sg^2_t$. According to Heston's model the time evolution $t \ra v_t$ follows the Cox-Ingersoll-Ross (CIR) process \cite{Shreve2004} 
\[
dv_t = -\g (v_t -\t) dt + \k \sqrt{v_t} dW_t^{(2)} \, .
\]
where $\t$ is the long-term mean of $v_t$, $\g > 0$ the rate of relaxation to this mean, $\k > 0$ the diffusion parameter influencing the noise level of the variance $v_t$, and $W^{(2)}_t$ is a standard Brownian motion. \\
We allow for a coupling between the two Brownian motions, that is,
\[
dW^{(2)}_t = \rho dW^{(1)}_t + \sqrt{1 - \rho^2}dZ_t \, .
\]
$Z_t$ is a Brownian motion independent of $W^{(1)}_t$, and $\rho \in [-1,1]$ is the instantaneous correlation between $W^{(1)}_t$ and $W^{(2)}_t$. A negative instantaneous correlation $\rho$ is known as the leverage effect \cite{FouqueJean-PierreGeorgePapanicolaou2000}.\\
While the CIR process is plausible, there exist other popular choices to describe the dynamics of the variance $v_t$. For example, the exponential Ornstein-Uhlenbeck model \cite{Masoliver2005} models the logarithm of the variance $y_t = \log \sg_t^2$ as an Ornstein-Uhlenbeck (OU) process, i.e., 
\[ 
dy_t = - \alpha (y_t - \mu) dt + \beta dW_t^{(2)}. 
\]
Under special conditions, the CIR process can be derived from an OU process: if $v_t$ is an CIR process then $\sqrt{v_t}$ is an OU process provided that $\gamma \theta = \frac{1}{4} \kappa^2$. \\
The Ornstein-Uhlenbeck process is an example of a wide class of stochastic processes, namely Gaussian processes. These have the property that for any finite collections of times $t_1, \ldots, t_n$, the observations of the process $y_{t_1}, \ldots, y_{t_n}$ have a multivariate Gaussian distribution. This in turn implies that the process is completely specified by its mean process $\mu_t = \mathbb{E}[y_t]$ and the covariance function $k_{t,t'} = \mathbb{E}[(y_t - \mu_t)(y_{t'} - \mu_{t'})]$. For the Ornstein-Uhlenbeck process, the stationary mean process is given by $\mu^{OU}_t \equiv \mu$ and the covariance is well known to be $k^{OU}_{t,t'} = \frac{\beta^2}{2 \alpha} e^{- \alpha |t - t'|}$. Below, we also consider other Gaussian processes thereby generalizing the class of stochastic volatility models beyond what is usually studied in finance.

\subsection{Heston's Model}
It is convenient to change the variable from price $S_t$ to the adjusted log-return $x_t = \log(S_t/S_0) - \mu  t$. Using It\^{o}'s formula \cite{Oksendal2000} we obtain 
\begin{equation} \label{eq:log_return}
d x_t = -\dfrac{v_t}{2} dt + \sqrt{v_t}dW^{(1)}_t \, .
\end{equation}
The two stochastic differential equations (SDE) for the variance $v_t$ and the adjusted log-return $x_t$ define a two-dimensional stochastic process
\begin{equation}\label{eq:3}
\left( \begin{array}{c}
x_t \\
v_t
\end{array} \right) = - \left( \begin{array}{c} 
v_t/2 \\
\g  (v -\t) 
\end{array} \right) dt
 + \sqrt{v_t} \left( \begin{array}{cc}
1 & 0 \\
\k \rho & \k \sqrt{1-\rho^2}
\end{array} \right) \left( \begin{array}{c}
dW^{(1)}_t \\
dZ_t
\end{array} \right) \, .
\end{equation}
The stochastic process $t \ra (x_t, v_t)$ possesses a joint probability density $p(x_t, v_t|v_0)$ with the initial condition $v = v_0$ for the variance. The initial condition of the SDE for the adjusted log-return $x_t$ is always $0$ by definition and therefore not mentioned explicitly. The time evolution of the joint probability density is governed by the Fokker-Planck or Kolmogorov forward equation
\begin{equation} \label{eq:1}
\partial_t p = \g \partial_v [(v-\t)p] + \dfrac{1}{2} \partial_x(vp)  + \rho\k \partial_{xv} (vp) + \dfrac{1}{2} \partial_{xx}(vp) + \dfrac{\k^2}{2} 
\partial_{vv} (vp) \, .
\end{equation}
The initial condition for this linear partial differential equation (PDE) is 
\begin{equation} \label{eq:2}
p(x_0, v_i|v_0) = \delta(x)\delta(v_i-v_0) \, .
\end{equation} 
The boundary conditions for the PDE \eq{1} at $v=0$ is governed by the boundary condition of the PDE for the probability density $\pi_t(v)$ of the variance $v_t$
\begin{equation} \label{eq:stat}
\partial_t \pi(v) = \partial_v [\g(v-\t)\pi(v)] + \dfrac{\k^2}{2}\partial_{vv}[v\pi(v)] \, .
\end{equation}
Feller \cite{Feller1951} carries out a detailed study of this parabolic PDE. For this process the classification of the boundary at the origin $v=0$ is as follows \cite{Karlin1981}: the origin is (instantaneously) reflecting, regular, and attainable if the Feller-constraint 
\[
\k^2 \leq 2\g\t
\]
is violated, and unattainable, non-attracting otherwise. Since we deal with parameter sets derived from market data only, Feller's constrained is reasonably fulfilled because variance of real stocks is always assumed greater than zero. Indeed, the parameter set derived in \cite{Ait-Sahalia2007} from the VIX provides $\a = 2\g\t/\k^2 = 2.011$. Hence, considering an attracting boundary of the PDE \eq{1} governed by the zero-flux condition \cite{Lucic2012} is not necessary, and we simply set $p(x_t,v=0|v_0)  = 0$ at the boundary $v=0$ for all time $t\geq 0$. In any case, the PDE \eq{2} has a stationary solution 
\begin{equation} \label{eq:variance}
\pi_\ast(v) = \dfrac{\a^\a}{\Gamma(\a)} \dfrac{v^{\a-1}}{\t^\a} e^{-\a v/\t}, \quad \a = \dfrac{2\g\t}{\k^2} \, 
\end{equation}
which is the Gamma distribution with shape $\a$ and rate $\b = \a/\t$. It has mean $\t$ and variance $\t/\a$. Feller's constraint implies $\a \geq 1$ and when $\a \ra \infty $ we obtain $\pi_\ast(v) \ra \delta(v-\t)$.
\begin{figure}[h] 
\includegraphics[width=\linewidth]{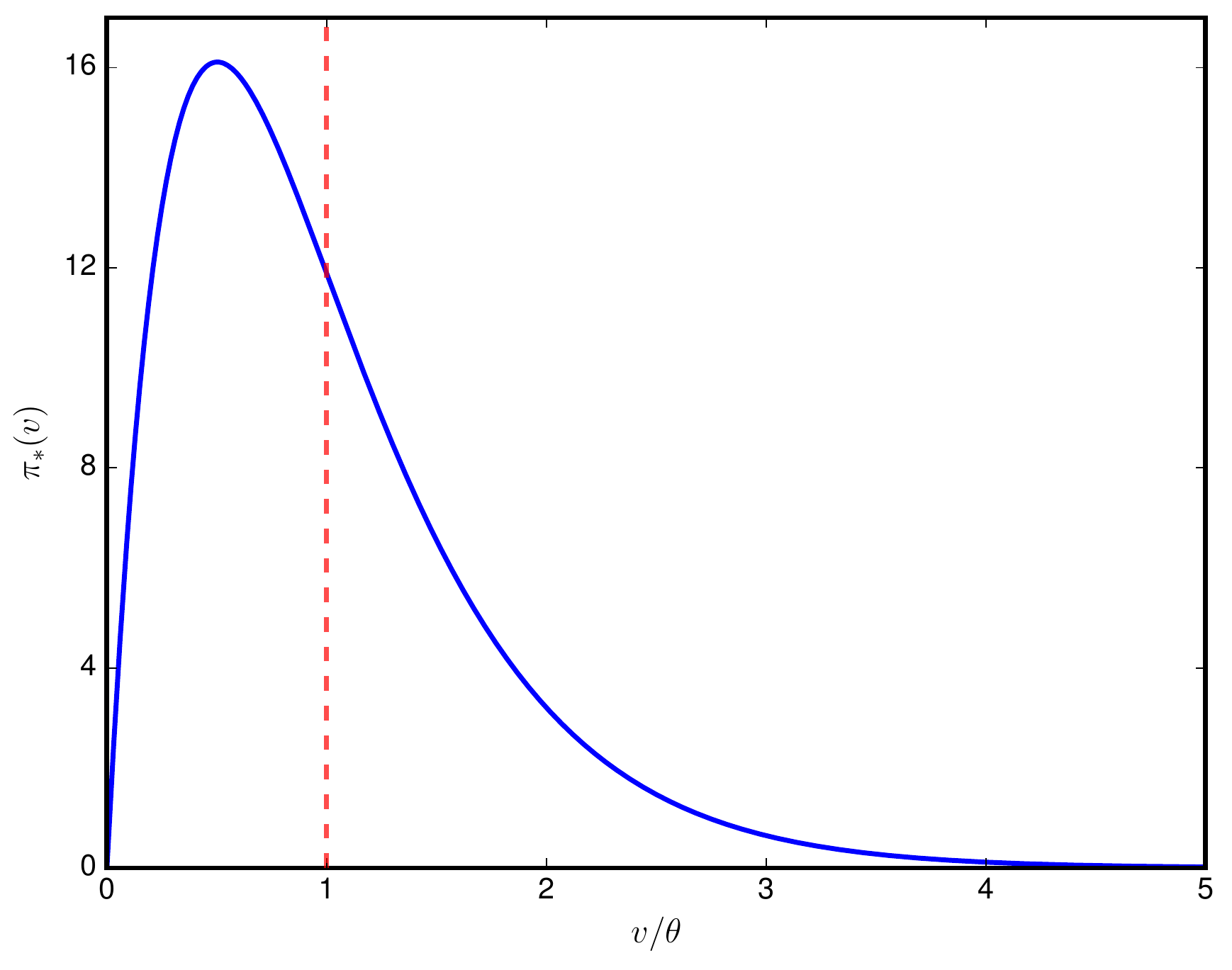}
\caption{\label{fig:fig_1} The stationary probability distribution $\pi_\ast(v)$ of the variance $v$ shown for $\a = 2.011$. The vertical line indicates the mean value of $v$.}
\end{figure}
Working with probability densities and their corresponding partial differential equations allows for a greater flexibility concerning the choice of initial conditions. For stochastic differential equations restricting the variance $v = v_0$ to a certain initial value $v_0$ at $t=0$ is the usual choice even though, in contrast to the returns $x_t$, we do not know the precise value $v_0$ of $v_t$ at $t = 0$. The choice of a random $v_0$ with distribution $\pi_\ast$ is more sensible because variance is
stationary distributed if Heston's model holds true for stock markets. Dealing with those initial conditions is barely possible within the SDE frame but setting 
\begin{equation} \label{eq:4}
p(x_0, v_0) = \delta(x_0) \pi_\ast(v_0)
\end{equation}
as initial condition for the joint probability density subject to the PDE \eq{1} is no problem at all and is adopted for our numerical simulation of the PDE \eq{1}. Since the PDE \eq{1} is linear, any linear combination of solutions is a solution as well. If $p(x_t,v_t | v_0)$ denotes a solution for the initial condition \eq{2}, the weighted average
\[
\int dv_0 \, p(x_t,v_t|v_0) \pi_\ast(v_0)
\]
is a solution of the PDE \eq{1} as well with initial value
\[
\int dv_0 \, p(x_0,v|v_0) \pi_\ast(v_0)  = \int dv_0\, \delta(x_0)\delta(v-v_0) \pi_\ast(v_0)= \delta(x_0)\pi_\ast(v_0) \, .
\]
Hence, solving the PDE \eq{1} for the initial condition \eq{4} yields as solution the weighted average of all solutions with known initial variance $v_0$.

\subsection{Solving the Fokker-Planck equation}
There exists an analytical solution of the Fokker-Planck equation \eq{1} in form of Fourier and inverse Laplace transforms which we shortly derive following \cite{Dragulescu2002}. First, we take the Fourier transform.
\[
p(x_t,v_t|v_0) = \int_{-\infty}^{\infty} \dfrac{dp_x}{2\pi} \exp(ip_x x_t)\bar{p}_{p_x}(v_t|v_0) \, ,
\]
i.e., $\bar{p}_{p_x}(v_t|v_0)$ is the Fourier transform of $p(x_t,v_t|v_0)$ w.r.t. $x_t$. Second, the Laplace transform 
\begin{equation} \label{eq:laplace}
\tilde{p}_{t, p_x}(p_v|v_0) = \int_{0}^{\infty} dv \exp(-p_v v) \bar{p}_{p_x}(v_t=v|v_0) \, .
\end{equation}
solves the partial differential equation of first order
\[
\left[ \partial_t + \left( \Gamma p_v + \dfrac{\k^2}{2}p_v^2 -  \dfrac{p_x^2-ip_x}{2} \right) \partial_{p_v} \right] \tilde{p} = -\g \t p_v \tilde{p}
\]
where we introduced the notation
\[
\Gamma = \g + i \rho \k p_x \,.
\]
This PDE has to be solved with initial condition
\[
\tilde{p}_{t=0, p_x}(p_v|v_0) = \exp(-p_v v_0) \, .
\]
The solution of this initial problem is given by the method of characteristics
\[
\tilde{p}_{t,p_x}(p_v|v_0) = \exp \left( -\tilde{p}_v(0)v_0 - \g \t \int_0^t d\tau \, \tilde{p}_v(\tau) \right) \, 
\]
where the function $\tilde{p}_v(\tau)$ is the solution of the characteristic (ordinary) differential equation
\[
\dfrac{\tilde{p}_v(\tau)}{d\tau} = \Gamma \tilde{p}_v(\tau) + \dfrac{\k^2}{2} \tilde{p}_v^2(\tau) - \dfrac{p_x^2 - \text{i}p_x}{2}
\]
with the boundary condition $\tilde{p}_v(t) = p_v$ specified at $\tau = t$. The differential equation is of the Riccati type with constant coefficients, and its solution is
\[
\tilde{p}_v(\tau) = \dfrac{2\Omega}{\k^2}\left( \zeta e^{\Omega(t-\tau)} - 1\right)^{-1} - \dfrac{\Gamma - \Omega}{\k^2} \, ,
\]
where we introduced the frequency 
\[
\Omega = \sqrt{\Gamma^2 +\k^2(p_x^2 -\text{i}p_x)}
\]
and the coefficient
\[
\zeta = 1 + \dfrac{2\Omega}{\k^2 p_v +(\Gamma - \Omega)} \, .
\]
Substituting yields
\begin{equation} \label{eq:solution}
\tilde{p}_{t, p_x}(p_v|v_0) =
 \exp \left( - \tilde{p}_v(0)v_0 + \dfrac{\a(\Gamma - \Omega) t}{2} - \a \log \dfrac{\zeta - e^{-\Omega t}}{\zeta -1}\right) \, 
\end{equation}
with $\a$ defined in \eq{variance}. Since we are interested in a solution of the Fokker-Planck equation w.r.t. to the initial conditions \eq{4} we take the average of all solution \eq{solution} with weight $\pi_\ast(v_0)$ and obtain
\begin{align} \label{eq:stat_sol}
\tilde{p}_{t, p_x}(p_v) &= \int_0^{\infty}dv_0 \, \tilde{p}_{t, p_x}(p_v|v_0)  \pi(v_0) \nn\\
 &= \left(1 +\dfrac{\t \tilde{p}_v(0)}{\a} \right)^{\a} \left( \dfrac{\zeta - e^{-\Omega t}}{\zeta -1} \right)^{\a} \exp \left( \dfrac{\a(\Gamma - \Omega) t}{2} \right) \, .
\end{align}
The solution $p(x_t,v_t)$ of \eq{1} with initial condition \eq{4} can be efficiently computed from \eq{stat_sol} applying the Stehfest method of degree $12$ -- see \cite{Cheng1994} -- to compute the inverse of the Laplace transform \eq{laplace}. This procedure performs magnitudes faster than solving the Fokker-Planck equation itself numerically, a method we tested also in order to validate our numerical findings by two unrelated computational approaches. Beside the joint distribution $p(x_t,v_t)$ we also need the conditional distribution 
\[
p(x_t|v_0) = \int_0^{\infty} dv_t \, p(x_t,v_t|v_0)
\]
which is the marginal distribution of the joint density $p(x_t,v_t|v_0)$, that is, the solution of the Fokker-Planck equation \eq{1} w.r.t. the initial conditions \eq{2}. There is a Fourier transform representation of the conditional density \cite{Dragulescu2002}
\begin{align}\label{eq:cond_joint}
p(x_t|v_0) = \int_{-\infty}^{\infty}\dfrac{dp_x}{2\pi} &\exp\left(\text{i} p_x x - 
\dfrac{p_x^2 - \text{i} p_x}{\Gamma  + \Omega \coth(\Omega t/2)}\right) \nn \\
\times&  \exp\left(-\dfrac{2\g\t}{\k^2}\log\left(\cosh \dfrac{\Omega t}{2} + 
\dfrac{\Gamma}{\Omega}\sinh\dfrac{\Omega t}{2}\right) + \dfrac{\g \Gamma \t t}{\k^2} \right) \, 
\end{align}
and the density of the marginal distribution of the returns \cite{Dragulescu2002}
\begin{equation}\label{eq:marginal}
p(x_t) = \int_{-\infty}^{\infty}\dfrac{dp_x}{2\pi} \exp\left(\text{i} p_x x + \dfrac{\g\theta}{\k^2}\Gamma t - \dfrac{2\g\theta}{\k^2} \log \left( \cosh \dfrac{\Omega t}{2} + \dfrac{\Omega^2 - \Gamma^2 + 2\gamma \Gamma}{2 \g \Omega}  \sinh \dfrac{\Omega t}{2} \right)  \right) \, .
\end{equation}
Knowledge of $p(x_t|v_0)$ and $p(x_t)$ allows for a numerical computation of the informational flow \eq{InfoFlow} via lemma \ref{lem_2} and of the mutual information $I(S_{(n+1)\tau}:S_{n\tau})$. All simulations are done with \textsc{Python} using \textsc{Numpy} \cite{Numpy2011}.

\section{Information Theory} 

\subsection{Mutual Information} The \textit{differential entropy} $h(X)$ of a continuous random real-variable $X$ with density $f(x)$ is defined as \cite{Cover2006}
\[
h(X) = - \int f(x) \log f(x) \, dx \, 
\]
if it exists, that is, the previous integral does not diverge. Since we measure information in bits we use the logarithm w.r.t. base $2$. In contrast to the discrete case, there is no canonical definition as one has to choose a reference measure $dx$ w.r.t. which the density is integrated. In the sequel we assume existence of the integrals without mentioning it. Further, we assume that the Lebesgue measure is chosen as reference. Compared to entropy of discrete random variables, differential entropy can be negative. One example provides a normally distributed random variable $X \sim (1/\sqrt{2\pi\sg^2})\exp(-x^2/2\sg^2)$ whose differential entropy is $1/2 \log(2\pi e\sg^2)$ \cite{Cover2006} which can be negative for sufficiently small variances $\sg^2$. Even though it might become negative, the differential entropy, as the entropy of discrete random variables, can be interpreted as a measure of the average uncertainty in the random variable. The differential entropy behaves nicely under diffeomorphic coordinate changes $\phi: S_X \ra \RH$ on the support $S_X = \set{x : f(x) > 0} \subseteq \RH$ of the random variable $X$
\[
h(\phi(X)) = h(X) + \int f(x) \log |\phi'(x)|  \, dx
\]
The \textit{differential entropy} of a set $X_1, X_2, \ldots, X_n$ of random variables with joint density $f(x_1, x_2,$ $ \ldots, x_n)$ is defined as
\[
h(X_1, X_2,  \ldots,  X_n) = - \int  f(x_1,x_2, \ldots, x_n) \log f(x_1, x_2, \ldots, x_n) \, dx_1dx_2 \cdots dx_n \, .
\]
For every differeomophism $\phi:S_\mathbf{X} \ra \RH^n$ on the support $S_\mathbf{X}= \set{\mathbf x = (x_1, x_2, \ldots, x_n):  \, f(x_1, x_2, \ldots, x_n) > 0} $ of the random variables $\mathbf{X} = (X_1, X_2, \ldots, X_n)$ we obtain
\[
h(\phi(\mathbf{X})) = h(\mathbf{X}) + \int f(\mathbf{x}) \log |\det J_\phi(\mathbf{x})| \, d\mathbf{x} 
\]
where $J_\phi (\mathbf{x})$ is the Jacobian of $\phi$.\\
If the random variables $X, Y$ have a joint density function $f(x,y)$ and conditional density function $f(x|y)$, respectively, we can define the \textit{conditional entropy} $h(X|Y)$ as
\[
h(X|Y) = - \int f(x,y) \log f(x|y) \, dx dy \, .
\]
Since in general $f(x|y) = f(x,y)/f(y)$, we can also write
\[
h(X|Y) = h(X,Y)-h(Y) \, .
\]
But we must be careful if any of the differential entropies are infinite. $h(X|Y)$ is a measure of the average uncertainty of the random variable $X$ conditional on the knowledge of another random variable $Y$. \\
The \textit{relative entropy} (or \textit{Kullback Leibler divergence}) $D(f||g)$ between two densities $f$ and $g$ is defined by 
\[
D(f||g) = \int f(x) \log \dfrac{f(x)}{g(x)} \,dx .
\]
Note that $D(f||g)$ is finite only if the support set $\set{x : f(x) > 0}$ of $f$ is contained in the support of $g$ (Motivated by continuity, we set $0 \log (0/0) = 0$.). While $D(f||g)$ is in general not symmetric, it is often considered as a kind of distance between $f$ and $g$. This is mainly due to its property that $D(f||g) \geq 0$ with equality if and only if $f = g$.\\
The \textit{mutual information} $I(X : Y)$ between two random variables $X$ and $Y$ with joint density $f(x, y)$ and respective marginal densities $g(x)$ and $k(y)$ is defined as 
\begin{eqnarray*}
I(X:Y) & = & h(X) - h(X|Y) \\
& = & h(Y) - h(Y|X) \\
& = & h(X) + h(Y) - h(X,Y) \\
& = & \int f(x,y) \log \dfrac{f(x,y)}{g(x)k(y)} \, dx dy \\
& = & D(f(x,y) || g(x)k(y)) \, .
\end{eqnarray*}
Again, we must be careful if any of the differential entropies are infinite and the mutual information might or might not diverge in this case. \\
From the definition it is clear that the mutual information is symmetric, i.e., $I(X:Y) = I(Y:X)$. Further, $I(X:Y) \geq 0$ with equality if and only if $X$ and $Y$ are independent. Thus, mutual information can be considered as a general measure of statistical dependence as it detects any deviations of independence. Note that in the case of independence, knowledge of $X$ does not reduce our uncertainty about $Y$, i.e., $X$ provides no information about $Y$, and vice versa. \\
Like mutual information, \textit{conditional mutual information} $I(X : Y|Z) $ between three random variables $X, Y$ and $Z$ can be written in terms of conditional entropies as
\[
I(X : Y|Z) = h(X|Z) - h(X|Y,Z)
\]  
assuming the differential entropies exist. This can be rewritten to show its relationship to mutual 
information
\[
I(X:Y|Z) = I(X: Y,Z) - I(X:Z)
\]
usually rearranged as the \textit{chain rule of mutual information} \cite{Cover2006}
\begin{equation} \label{eq:chain_rule}
I(X:Y,Z) = I(X:Y|Z)+I(X:Z) \, . 
\end{equation}
In contrast to differential entropy, mutual information and conditional mutual information are scaling invariant, that is, for three diffeomorphic maps $\phi_X, \phi_Y, \phi_Z: \RH \ra \RH$ we have \cite{Kraskov2004}
\begin{align*}
I(\phi_X(X) : \phi_Y(Y)&) = I(X:Y) \\
I(\phi_X(X) : \phi_Y(Y)&| \phi_Z(Z)) = I(X:Y|Z).
\end{align*}

In the last section we defined
\begin{align*}
x_t &= \log(S_t/S_0) - \mu t \\
v_t &= \sg^2_t \, .
\end{align*}
Both transformations are diffeomorphic on the support of the respective random variables. Scaling invariance of the mutual information yields

\begin{lem}\label{lem_1}
The mutual information between stock and volatility is the same as the one between adjusted log-returns and variance.
\[
I(\sg_t:S_t) = I(v_t:x_t)
\]
\end{lem}

\subsection{Multilevel Dynamical Systems}
We consider the process $t \ra (S_t, \sg_t)$ with a fixed initial
value $S_0$, the stock price today, and a stationary distributed
volatility $\sg_0$. Since volatility is a hidden parameter, the only
observable is the stock $S_t$. 
Then, we discretize the continuous processes $S_t$ and $\sg_t$ in time
with resolution $\tau$. This is sensible because the observable, that
is, stock data, is in general only quoted at discrete time points,
most commonly at the end of every trading day, i.e., $\tau = 1$
day. From this we obtain a Markov process in discrete time $n \ra
(S_{n\tau}, \sg_{n\tau})$. Since only the stock price can be observed,
we can illustrate the situation as in \fig{levels}
with $\phi$ being the observation map that simply projects onto its
first component. \\
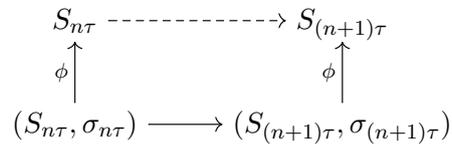
\begin{figure}[h]
  \begin{center}
    \begin{tikzcd}
      S_{n\tau}\arrow[dashed]{r} & S_{(n+1)\tau} \\
      (S_{n\tau},\sg_{n\tau}) \arrow{r}\arrow{u}{\phi} & (S_{(n+1)\tau}, \sg_{(n+1)\tau})\arrow{u}{\phi}
    \end{tikzcd}
  \end{center}
  \caption{\label{fig:levels} Basic setup of a low level process that
    can be observed via an observation map $\phi$.}
\end{figure}
We investigate to which extend the observed process $n \ra S_{n\tau}$
is a stochastic process in its own right.  The dashed line on the
upper level indicates that the upper process is not self-contained,
because $S_{(n+1)\tau}$ cannot in general be derived from $S_{n\tau}$
solely, whereas the solid line on the lower level indicates the
Markovian dynamics for the stochastic differentials for $S_{n\tau}$
and $\sg_{n\tau}$ or, equivalently, $x_{n\tau}$ and $v_{n\tau}$
via \eq{3}. Such multilevel dynamical systems are extensively studied in \cite{Pfante2014} where we introduced various information theoretical measures in order to make the notion \lgans self-contained" precise and the deviation of a process from being self-contained quantifiable. Two of them are of particular interest.\\
\textbf{Informational closure:} We called the higher process to be
informationally closed, if there is no information flow from the lower
to the higher level. Knowledge of the joint state $(S_{n\tau},
\sg_{n\tau})$ will not improve predictions of the stock $S_{(n+1)\tau}$
at time $n+1$ when $S_{n\tau}$ is already known, i.e., for
$S_{(n+1)\tau}=\phi(S_{(n+1)\tau},\sg_{(n+1)\tau})$ we have
 \begin{equation} \label{eq:InfoFlow}
   I(S_{(n+1)\tau} : (S_{n\tau}, \sg_{n\tau}) | S_{n\tau}) = 0 \, .
 \end{equation}
 \textbf{Markovianity:} Markovianity of the upper process $n \ra S_{n\tau}$ is another property that allows to be considered a self-contained process in its own right. In this case $S_{(n+1)\tau}$ is independent of the entire past trajectory 
\[
 \vec{S}^\tau_{n-1}=(S_{(n-1)\tau}, \ldots, S_{\tau},  S_{0}) \, 
\]
given $S_{n\tau}$, which can be expressed again in terms of the conditional mutual information as
\begin{equation} \label{eq:true_markov} I(S_{(n+1)\tau}
  :\vec{S}^\tau_{n-1} | S_{n\tau}) = 0 \, .
\end{equation}

In our setting, i.e., with $S_t$ and $\sg_t$ derived from Heston's
model with initial condition \eq{2} and then discretized according to
\fig{levels} we obtain the following results:

\begin{lem} \label{lem_2} The information flow can be computed as
  \[
  I(S_{(n+1)\tau} : (S_{n\tau}, \sg_{n\tau}) | S_{n\tau}) = I(x_{(n+1)\tau} : v_{n\tau} | x_{n\tau})
  \]
  for all $\tau > 0$ and $n \in \NH$.
\end{lem}
\begin{proof}
  Lemma \ref{lem_1} implies
  \[
  I(S_{(n+1)\tau} : (S_{n\tau}, \sg_{n\tau}) | S_{n\tau}) =
  I(x_{(n+1)\tau} : (x_{n\tau}, v_{n\tau}) | x_{n\tau})
  \]
  Then, by the chain rule of mutual information \eq{chain_rule} this
  can be expanded as
  \[
    I(x_{(n+1)\tau} : (x_{n\tau}, v_{n\tau}) | x_{n\tau}) = I(x_{(n+1)\tau} : v_{n\tau} | x_{n\tau}) + I(x_{(n+1)\tau} : x_{n\tau}| x_{n\tau}, v_{n\tau}) \, .
  \]
  The result then follows from $I(x_{(n+1)\tau} : x_{n\tau}|
  x_{n\tau}, v_{n\tau}) = 0$.
\end{proof}

As a corollary, for the first step, i.e., $n=0$ we obtain

\begin{cor}
\[
I(S_\tau: (S_0,  \sg_0)|S_0) = I(x_\tau:v_0)
\]
\end{cor}

from lemma \ref{lem_2} and by dropping the implicit condition $x_0=0$. Similarly, we can derive a bound on the deviation from Markovianity for the stock process $n \ra S_{n\tau}$ based on the information flow.

\begin{thm} \label{lem_3}
In the setting of \fig{levels}
\[
I(S_{(n+1)\tau} :\vec{S}^\tau_{n-1}  | S_{n\tau}) \leq I(x_{(n+1)\tau} : v_{n\tau} | x_{n\tau}) \, .
\]
for all $\tau > 0$ and $n \in \NH$.
\end{thm}
\begin{proof}
  Scale invariance and decomposing the mutual information yields
  \begin{align*}
    I(S_{(n+1)\tau} :\vec{S}^\tau_{n-1}  | S_{n\tau}) =& I(x_{(n+1)\tau} :\vec{x}^\tau_{n-1}  | x_{n\tau}) \\
    \leq& I(x_{(n+1)\tau} : (\vec{x}^\tau_{n-1}, v_{n\tau})  | x_{n\tau}) \\
    =& I(x_{(n+1)\tau} : v_{n\tau}  | x_{n\tau}) + I(x_{(n+1)\tau} : \vec{x}^\tau_{n-1}  | x_{n\tau}, v_{n\tau})
  \end{align*}
  Markovianity of the process $n \ra (x_{n\tau},v_{n\tau})$ yields $I(x_{(n+1)\tau} : \vec{x}^\tau_{n-1}  | x_{n\tau}, v_{n\tau}) = 0$.
\end{proof}

\section{Inferring volatility from stock data}

\subsection{Information in the Heston Model}
As an example for our numerical studies we take realistic parameters from \cite{Ait-Sahalia2007} where Heston's model was fitted directly on the S{\&}P 500 and its VIX. They got the values
\begin{equation}
\label{eq:params}
\g = 5.07, \, \t = 0.0457, \k = 0.48, \rho = -0.767
\end{equation}
The VIX is quoted in percentage points and translates, roughly, to the expected 
movement (with the assumption of a 68{\%} likelihood, i.e., one standard deviation) in 
the S{\&}P 500 index over the next 30-day period, which is then annualized. Hence, the solution of the PDE \eq{solution} with these parameters and initial condition \eq{4} until $t = 2.0$ corresponds to two years, that is, $554$ trading days.

\begin{figure}[h] 
\includegraphics[width=\linewidth]{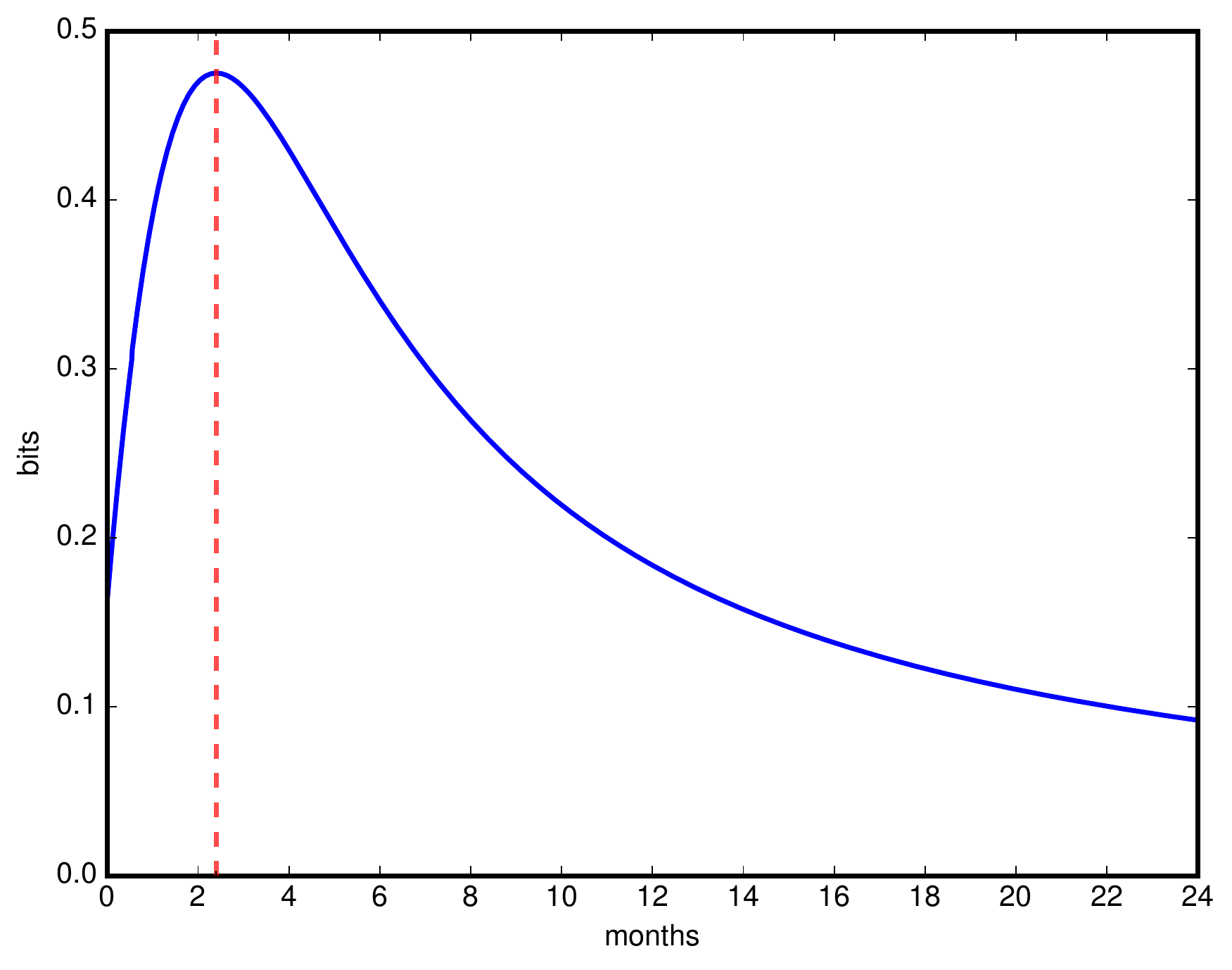}
\caption{\label{fig:fig_2} The evolution of the mutual information $I(\sg_t : S_t)$ for the joint distribution $p(x_t,v_t)$ solving the PDE \eq{1} with initial condition \eq{4}. The maximum is approximately at $t = 55$ trading days. } 
\end{figure}

Originally, we computed the mutual information $I(v_t:x_t)$ but due to lemma \ref{lem_1}
we have $I(\sg_t:S_t)= I(v_t:x_t)$. One can read off \fig{fig_2} three remarkable facts. 
First, the maximum $I(\sg_\tau:S_\tau) = 0.48$ bits at 
$\tau = 55$ trading days. Second, the mutual information diminishes for longer times $t$.
Third, and most important, the small amount of information left in the stock 
about its underlying volatility. We get at most about half a bit which makes at least $25$ independent
measurements necessary to derive $10$ bits, that is about three digits of the volatility. 
We could possibly obtain more information by observing multiple successive stock prices $S_t, S_{t-1}, \ldots$, but due to the stochasticity of the volatility process the information is still rather limited. Below, we illustrate this effect in the case of stochastic volatility models based on Gaussian processes. Third, the instantaneous jump of the mutual information $I(v_t:x_t)$ for $t \ra 0$. Since the distributions of the adjusted log-return $x_t$ and the variance $v_t$ are chosen to be independent at $t=0$, see \eq{4}, mutual information $I(x_0:v_0) = 0$. But the numerical simulation suggest the limit $\lim_{t\ra 0} I(x_t:v_t) \approx 0.166$ bits. \\
Next, we compute the informational flow.
\[
I(S_{(n+1)\tau} : (S_{n\tau}, \sg_{n\tau}) | S_{n\tau}) = I(S_{(n+1)\tau} :  \sg_{n\tau} | S_{n\tau}) 
\]
Lemma \ref{lem_2} yields
\begin{align} \label{eq:misc}
I(S_{(n+1)\tau} :  \sg_{n\tau} | S_{n\tau}) &= I(x_{(n+1)\tau} : v_{n\tau}| x_{n\tau})\nn\\
 &= h(x_{(n+1)\tau} |x_{n\tau}) - h(x_{(n+1)\tau} | x_{n\tau},v_{n\tau})\nn \\
&= h(x_{(n+1)\tau} , x_{n\tau}) - h(x_{n\tau}) - h(x_{\tau} |x_ 0, v_0) \, 
\end{align}   
where the last equality follows from $p(x_{(n+1)\tau}| x_{n\tau}, v_{n\tau}) = p(x_{\tau}|x_0,v_0)$, i.e., the temporal homogeneity of solutions of the Fokker-Planck equation \eq{1}. The conditional distribution $p(x_{\tau} |x_0, v_{0})$ is in general not equal \eqref{eq:cond_joint} as long as the condition on $x_0$ is replaced by $x_0=0$. If $x_0 = x$, then
\[
h(x_{\tau} |x_ 0, v_0) = - \int p(x_{\tau} =y| x_0=x, v_{0}) \log  p(x_{\tau}=y | x_0=x, v_{0})  \, dy
\]      
If we change the coordinates $y' = y-x$ and $x' = x$, the Jacobian of this coordinate transformation is $1$ and we obtain
\[
h(x_{\tau} |x_ 0, v_0) = - \int p(x_{\tau} =y'+x'| x_0=x', v_{0}) \log  p(x_{\tau}=y'+x' | x_0=x', v_{0})  \, dy'
\]
Since solutions $(x_t, v_t, x_0, v_0) \mapsto p(x_t, v_t|x_0, v_0)$ of the PDE \eq{1} depend only on  $(x_t-x_0, v_t, 0 ,v_0)$, we obtain $ p(x_{\tau}=y'+x' | x_0=x', v_{0} ) =  p(x_{\tau}=y' | x_0=0, v_0)$. This implies $h(x_{\tau} |x_ 0, v_0) = h(x_{\tau} |v_0)$ adopting the rule that we drop the implicit condition $x_0 = 0$. Hence, the entropy $h(x_{\tau} |x_ 0, v_0)$ is the one of the conditional density \eq{cond_joint}.\\
Computing the joint density $p(x_{((n+1)\tau}, x_{n\tau})$ yields
\begin{align*}
p(x_{(n+1)\tau} , x_{n\tau}) &= \int p(x_{(n+1)\tau} , x_{n\tau}, v_{n\tau} = v) dv \\
&= \int p(x_{(n+1)\tau} | x_{n\tau}, v_{n\tau} = v) p( x_{n\tau}, v_{n\tau} = v) dv\\
&= \int p(x_{\tau} | x_0, v_{0} = v) p( x_{n\tau}, v_{n\tau} = v) dv \, .
\end{align*}
Its entropy reads
\begin{align*}
h(x_{(n+1)\tau} , x_{n\tau}) &= - \int p(x_{(n+1)\tau}=y , x_{n\tau}=x)  \log p(x_{(n+1)\tau}=y , x_{n\tau}=x) \, dxdy \\
&= - \int \bigg(\int p(x_{\tau} =y| x_0=x, v_{0} = v)  p( x_{n\tau}=x, v_{n\tau} = v) dv \bigg)\\
 & \quad \times \log \bigg(\int p(x_{\tau}=y | x_0=x, v_{0} = v) p( x_{n\tau}=x, v_{n\tau} = v) dv  \bigg)\, dxdy \, . \\
\end{align*}
Furthermore, changing coordinates $y' = y-x$ and $x' = x$ yields
\begin{align*}
h(x_{(n+1)\tau}& , x_{n\tau}) =\\
  -&\int \bigg(\int p(x_{\tau} =y'+x'| x_0=x', v_{0} = v)  p( x_{n\tau}=x', v_{n\tau} = v) dv \bigg)\\
 & \times \log \bigg(\int p(x_{\tau}=y'+x' | x_0=x', v_{0} = v)   p( x_{n\tau}=x', v_{n\tau} = v) dv\bigg) \, dx'dy' \, . \\
\end{align*}
We substitute $ p(x_{\tau}=y'+x' | x_0=x', v_{0} = v) =  p(x_{\tau}=y' | x_0=0, v_{0} = v)$ and reduced \eq{misc} to expressions which can be computed via \eq{cond_joint} and \eq{solution}, at least numerically. We simulated the informational flow \eq{InfoFlow} for several choices of $n$.   \\
Beside the informational flow, we are also interested in the mutual information $I(S_{(n+1)\tau} : S_{n\tau}) = I(x_{(n+1) \tau}:x_{n\tau})$ between two subsequent observations of the stock. Definition of the conditional mutual information yields
\begin{align*}
I(x_{(n+1)\tau} : v_{n\tau} | x_{n\tau} ) &= I(x_{(n+1)\tau} : v_{n\tau},  x_{n\tau}) - I(x_{(n+1)\tau} : x_{n\tau}) \\
&= h(x_{(n+1)\tau} ) - h(x_{(n+1)\tau} | v_{n\tau},  x_{n\tau}) - I(x_{(n+1)\tau} : x_{n\tau})\\
&= h(x_{(n+1)\tau} ) - h(x_{\tau} | v_{0},  x_{0}) - I(x_{(n+1)\tau} : x_{n\tau})\\
&= h(x_{(n+1)\tau} ) - h(x_{\tau} | v_{0}) - I(x_{(n+1)\tau} : x_{n\tau})
\end{align*}
where the last equation follows from the previously proven identity $h(x_{\tau} | v_{0}, x_0) = h(x_{\tau} | v_{0})$ where the last entropy is the one of the conditional distribution \eq{cond_joint}. The entropy $h(x_{(n+1)\tau} ) $ can be computed via formula \eq{marginal} for the density of $x_t$. Hence,
\begin{equation} \label{eq:mutual_info_stock}
I(S_{(n+1)\tau} : S_{n\tau}) = I(x_{(n+1)\tau} : x_{n\tau}) = h(x_{(n+1)\tau} ) - h(x_{\tau} | v_{0}) - I(x_{(n+1)\tau} : v_{n\tau} | x_{n\tau} ) \,.
\end{equation}
In particular, if we set $n = 0$ we obtain $I(S_{\tau} : S_{0}) = 0$. For the other values of $n$ we plotted \fig{ratio} the ratio
\begin{equation}\label{eq:ratio}
\dfrac{I(S_{(n+1)\tau} : \sg_{n\tau} | S_{n\tau} )}{I(S_{(n+1)\tau} : S_{n\tau})}
\end{equation}
which expresses the informational gain on predictions of future stock movements $S_{(n+1)\tau}$ we obtain from knowing the instantaneous volatility compared to the information already there from reading off stock data $S_{n\tau}$ at time $n\tau$.
\begin{figure}[h] 
\includegraphics[width=\linewidth]{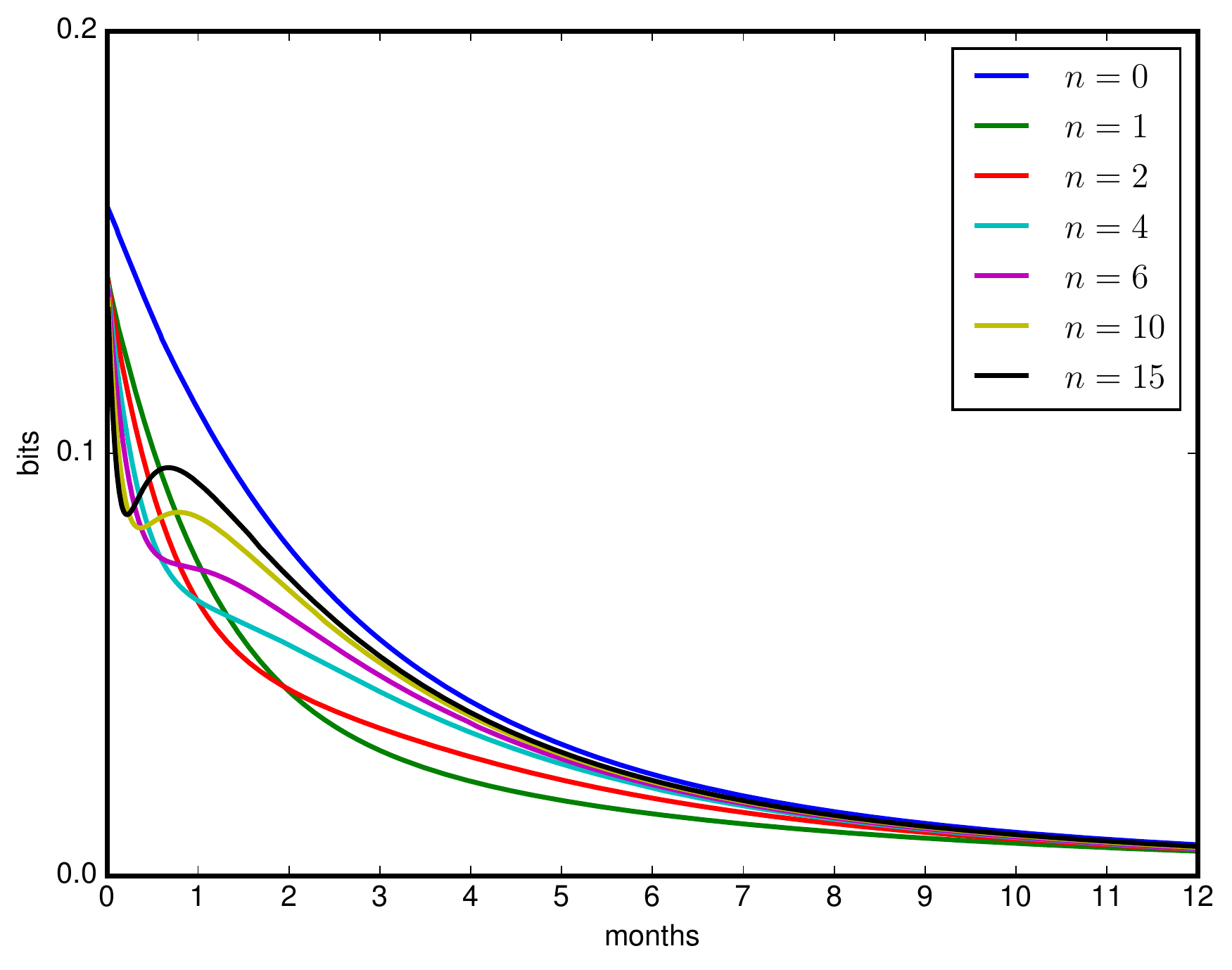}
\caption{\label{fig:info} The informational flow $I(S_{(n+1)\tau} :  \sg_{n\tau} | S_{n\tau})  $ for $n = 0,1,2,4,6,10$, and $15$. } 
\end{figure}
\begin{figure}[h] 
\includegraphics[width=\linewidth]{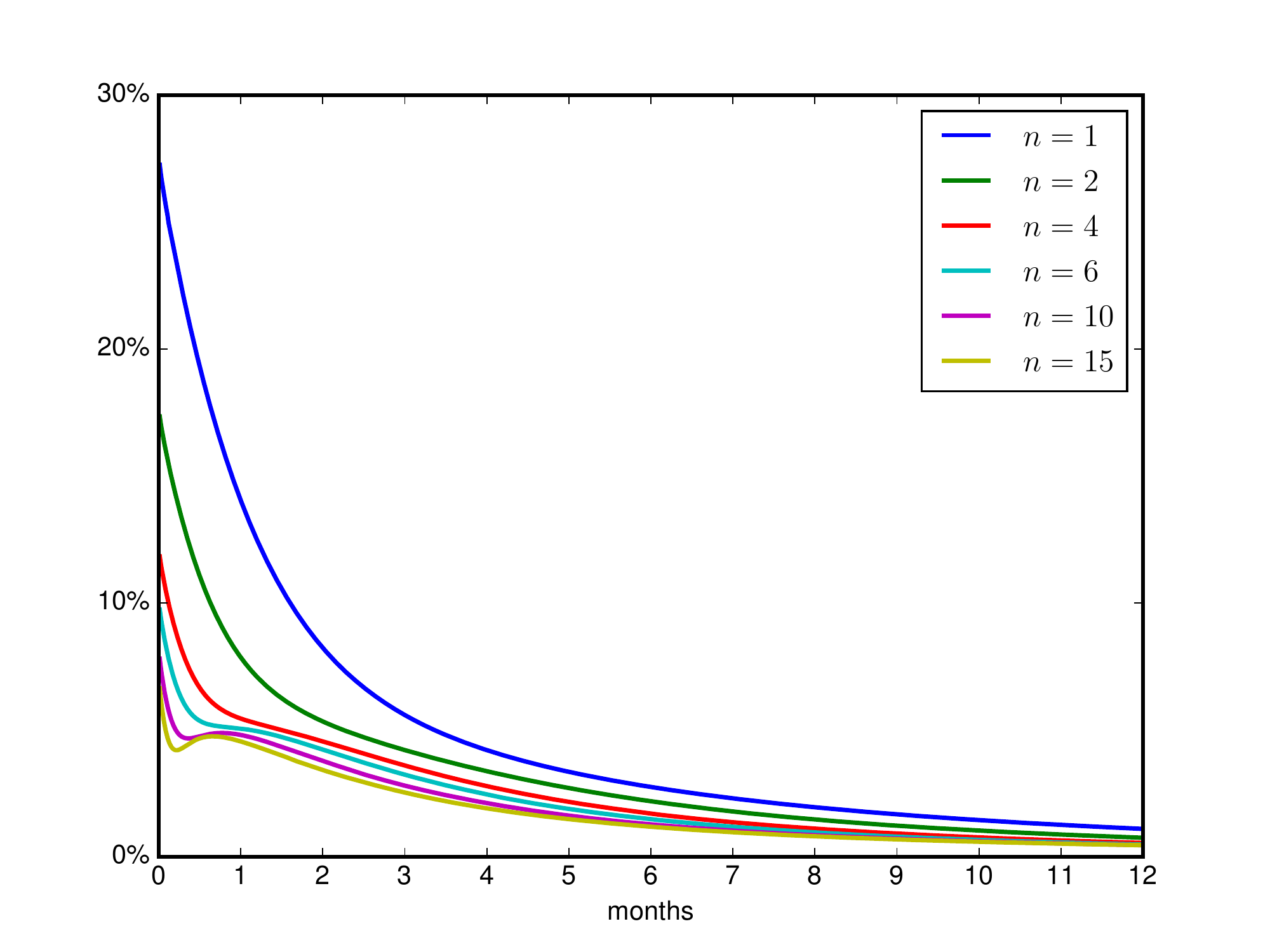}
\caption{\label{fig:ratio} The ratio $I(S_{(n+1)\tau} :  \sg_{n\tau} | S_{n\tau})  / I(S_{(n+1)\tau} :  S_{n\tau}) $ for $n = 1,2,4,6,10$, and $15$. } 
\end{figure}

Several facts can be read off \fig{info} and \fig{ratio}. First, the informational flow $I(S_\tau: \sg_0|S_0) = I(x_\tau : v_0)$ is an upper bound of the general flow $I(S_{(n+1)\tau} :  \sg_{n\tau} | S_{n\tau})  = I(x_{(n+1)\tau} : v_{n\tau} | x_{n\tau}) $ for all $n \in \NH_0$. Second, the incredible small amount of this flow: the knowledge of the joint state $(S_{n\tau}, \sg_{n\tau}) $ will not remarkably improve predictions on the stock $S_{(n+1)\tau}$ at time $\tau$ if $S_{n\tau}$ is already known. The informational gain $I(S_{(n+1)\tau} :  \sg_{n\tau} | S_{n\tau})$ compared to the information $I(S_{(n+1)\tau} :  S_{n\tau}) $ we already have about $S_{(n+1)\tau} $ if we read off $S_{n\tau}$ is at most $27\%$ and vanishes if $\tau$ or $n$ increases. Diminishing informational gain on future returns from knowing the instantaneous volatility $v_{n\tau}$ if $n$ grows can be directly derived from \eq{mutual_info_stock}: we already know from \fig{info} that $I(S_{(n+1)\tau} :  \sg_{n\tau} | S_{n\tau})$ is upper bounded by $I(S_{\tau} :  \sg_{0} | S_{0}) = h(x_\tau)-h(x_\tau|v_0)$ independent of $n$ whereas $h(x_{(n+1)\tau})$ is extensive in $n$ because the marginal distribution of $x_{(n+1)\tau}$ is roughly Gaussian with a variance proportional to $(n+1)\tau$. Hence, $h(x_{(n+1)\tau}) \approx \mathcal{O}\left(\log \left( (n+1)\tau \right)\right)$. The same argument yields $h(x_{\tau}|v_0) \approx \mathcal{O}(\log \left(\tau)\right)$ and therefore
\[
I(S_{(n+1)\tau} : S_{n\tau}) = I(x_{(n+1)\tau} : x_{n\tau}) \approx \log \left( (n+1)\tau\right) - \log \tau = \log(n+1)\, ,
\]
that is, the denominator in \eq{ratio} grows logarithmically in $n$. Overall, knowing the instantaneous volatility $\sg_{n\tau}$ does not deliver considerably better estimations on future stock-returns than the knowledge of the stationary distribution $\pi_\ast(\sg_n)$ of the volatility does, which we know already from the model. Therefore, inferring the realized volatility from market data in financial markets subject to Heston's model is not only hard, but also of limited value when predicting prices. Third, from theorem \ref{lem_3} follows that the stock process $S_{n \tau} \ra S_{(n+1)\tau}$ is nearly Markovian for increasing $\tau$. Last, also the informational flow $I(S_{(n+1)\tau} :  \sg_{n\tau} | S_{n\tau}) $ makes an instantaneous jump for $\tau \ra 0$ with the same jump height for $n=0$ as the mutual information $I(x_\tau: v_\tau)$ in \fig{fig_2}.\\
Finally, we investigate in this subsection the influence of the parameters $\rho$ on the mutual information $I(\sg_t, S_t)$. Recall, $\rho$ is the instantaneous correlation between the two Brownian motions $W^{(1)}_t$ and $W^{(2)}_t$ driving the processes
$S_t$ and $v_t$, respectively. $\rho$ is negative to capture the leverage effect \cite{Bouchaud2001}, i.e., negative returns tend to increase volatility. The effect on the joint probability is nicely captured by figure \fig{fig_3}. The joint distribution of the first figure is skewed towards negative returns if volatility increases whereas the second presents a nearly symmetric distribution, i.e., negative returns affect volatility not more than positive ones.

\begin{figure}
  \begin{tabular}{cc}
 Leverage Effect Included & Leverage Effect Excluded \\
    \includegraphics[width=0.48\linewidth]{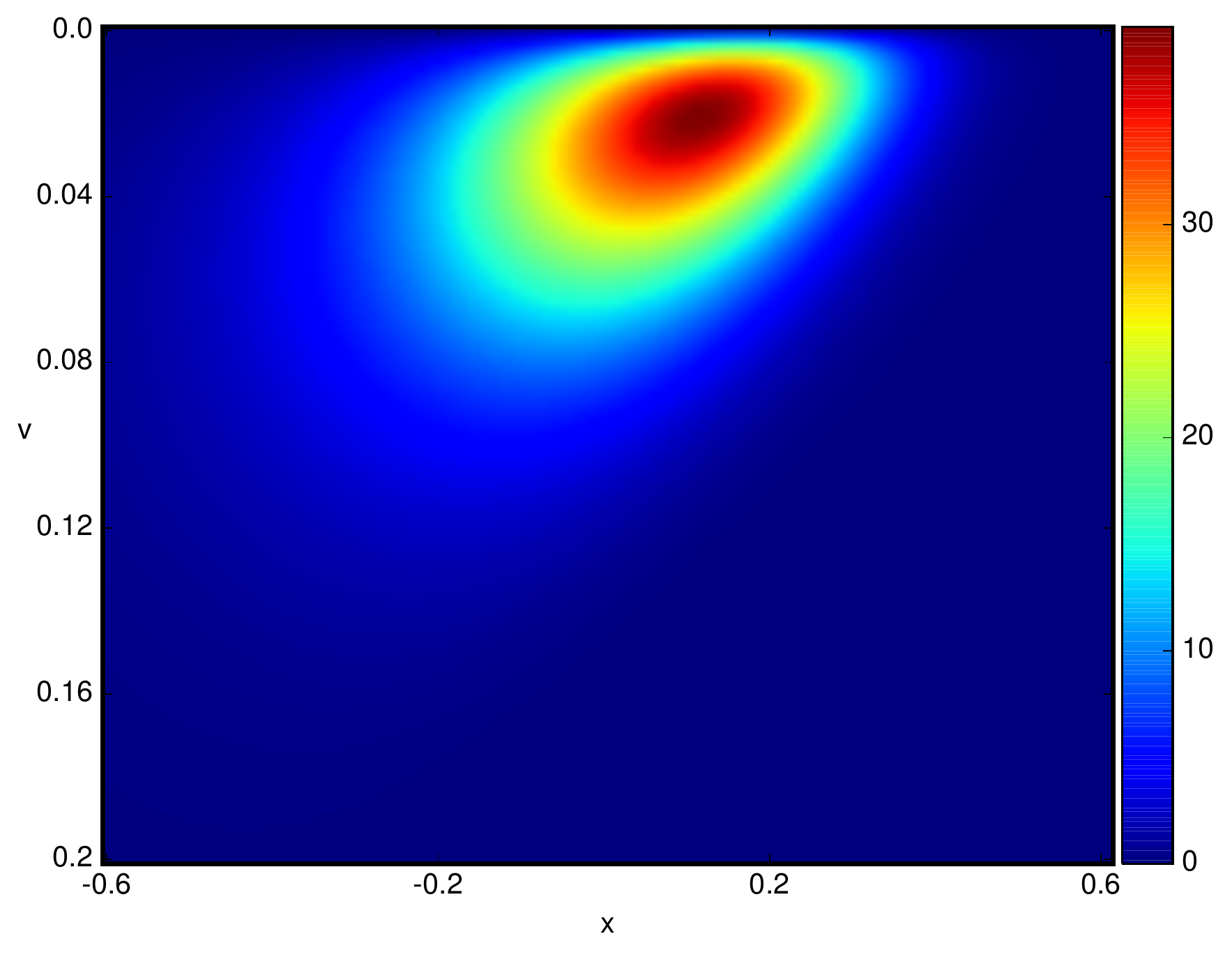} &
    \includegraphics[width=0.48\linewidth]{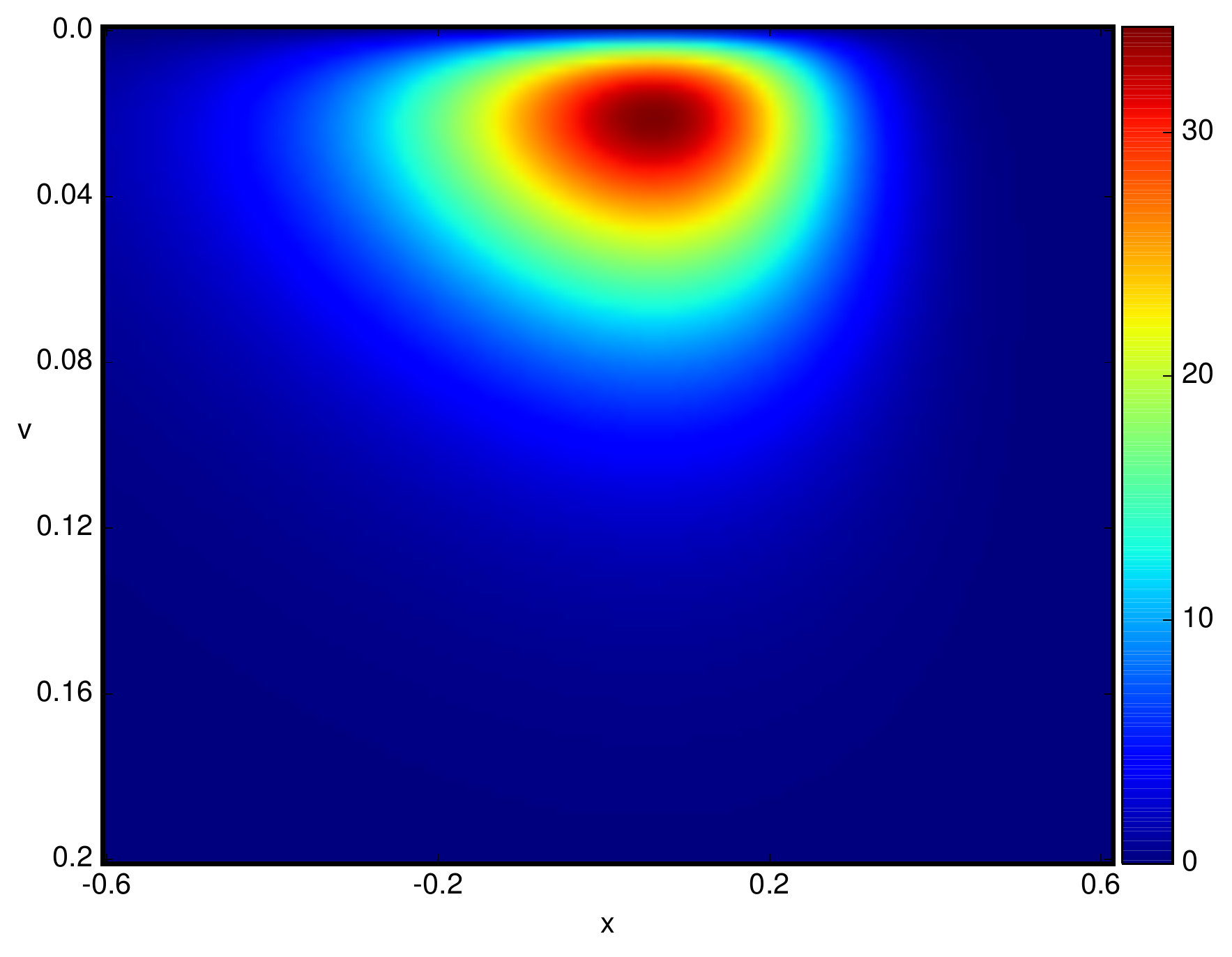}
  \end{tabular}
  \caption{\label{fig:fig_3}The joint distribution $p(x_\tau, v_\tau)$ of
  the Heston process with parameters as in \eq{params} for $\tau =
  252$ trading days. If $\rho = -0.767$, the inclination shows the negative correlation
  between adjusted log-returns $x$ and variance $v$. If $\rho=0.0$ the distribution is nearly symmetric in the adjusted
  log returns $x$.}
\end{figure}

Figure \fig{fig_5} shows the mutual information $I(\sg_t:S_t)$ for different choices of the instantaneous
correlation $\rho$. Two facts are remarkable. The growth of the information w.r.t. to $\rho$ and the negligible amount left if we choose $\rho=0.0$. In this case the stock and the volatility process are quasi independent and nearly no information about the instantaneous volatility can be obtained from stock data.

\begin{figure}[h] 
\includegraphics[width=\linewidth]{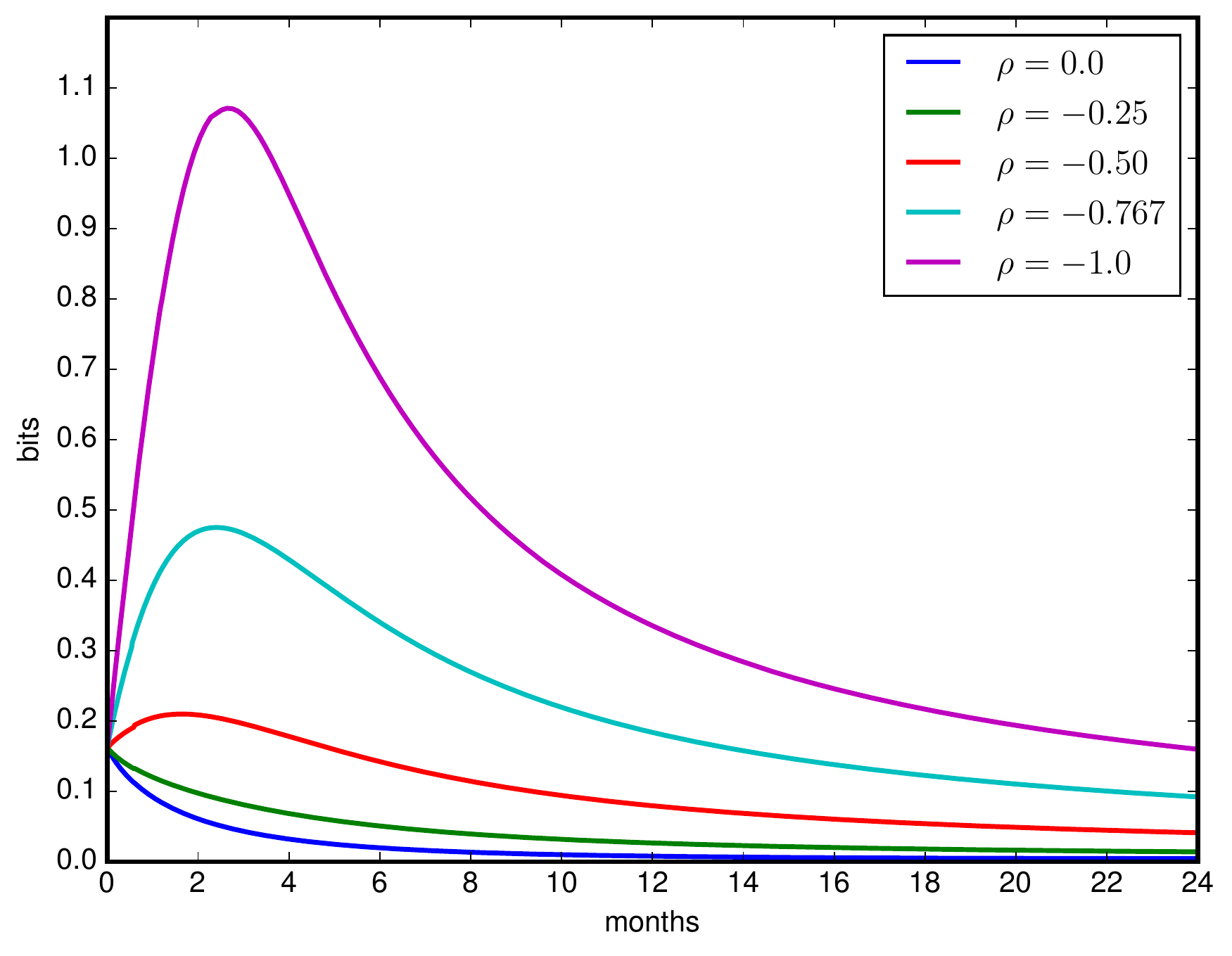}
\caption{\label{fig:fig_5} The evolution of the mutual information $I(\sg_t : S_t)$ for
different values of the instantaneous correlation $\rho$. } 
\end{figure}

\subsection{Fitting Stochastic Volatility Models}

The mutual information $I(S_t : \sg_t)$ quantifies the average
dependence between the stock price and instantaneous volatility which
we have evaluated for the Heston model above. When inferring
volatility from actual stock price data we are usually less interested
in the average behaviour over possible realizations of the price
process $t \ra S_t$. Rather, we want to estimate the volatility
conditioned on the actual sequence of prices $\vec{S} = (S_{t_1},
\ldots, S_{t_n})$ that were observed at times $t_1 < \cdots < t_n$.
Unfortunately, computing the full conditional distribution $p(\sg_t
\mid \vec{S})$ or even $p(\{\sg_t\}_{t>0} : \vec{S})$ is infeasible in the
Heston model.  Thus, in this section we turn to a different class of models which is based on
Gaussian processes. While exact inference is also infeasible in these
models, Gaussian processes are heavily used in machine learning and a wide
range of approximation algorithms is readily available. As described
above, we consider generalizations of the exponential Ornstein-Uhlenbeck
model. As before, the price process is given by
\[ 
dS_t= \mu S_t dt + \sg_t S_t dW_t^{(1)},
\]
but now, the logarithm of the variance $y_t = \log \sg^2_t$ is drawn
from a Gaussian process. To fit such a model, we can represent the
distribution of $\vec{y} = (y_{t_1}, \ldots, y_{t_n})$ exactly by a
multivariate Gaussian distribution. In addition, we need the
likelihood of the observed prices $\vec{S}$ given $\vec{y}$, i.e., $p(\vec{S} \mid \vec{y})$. Here,
for simplicity, we use an Euler approximation to the stochastic
differential equation for $S_t$ \cite{Phillips2009}:
\[ S_{t+\Delta t} - S_t = \mu S_t \Delta t + \sigma_t S_t \sqrt{\Delta
  t} \epsilon_t \] where $\epsilon_t \sim \mathcal{N}(0,1)$.

Thus, we assume that the observed daily return $r_t = \frac{S_{t+1} -
  S_{t}}{S_t}$ is drawn from a normal distribution with standard
deviation $\sigma_t$, i.e., $r_t \sim \mathcal{N}(\mu,
\sigma_t)$. Since the drift of the price process is usually rather
small compared to its volatility, we will assume $\mu = 0$ for
simplicity in the following.

In this formulation, our stochastic volatility model is a Gaussian
process $y_t$ with a non-Gaussian likelihood model $p(r_t |
y_t)$. Similar models are well known and widely used in Gaussian
process classification (see chapter 3 of \cite{Rasmussen2006}). In this context, several algorithms
to approximate the posterior $p(\vec{y} | \vec{r})$ have been developed. Here, we employ the Laplace
approximation where the posterior is approximated by a multivariate
Gaussian centred at the mode of the posterior, i.e.,
$\vec{y}_{\mathrm{Laplace}} = \argmax_{\vec{y}}p(\vec{y} \mid
\vec{r}) = \argmax_{\vec{y}} \Psi(\vec{y})$ where $\Psi(\vec{y}) = \log p(\vec{r}|\vec{y}) p(\vec{y})$. The covariance for the approximate posterior is derived
from the local curvature around the mode,
i.e., $\Sigma_{\mathrm{Laplace}}^{-1} = - \nabla \nabla \log
p(\vec{y}_{\mathrm{Laplace}} \mid \vec{r} )$. 

Using the Laplace approximation it is also possible to compute an
approximation to the marginal likelihood
\begin{eqnarray*} 
  p(\vec{r}) & = & \int p(\vec{r} \mid \vec{y}) p(\vec{y}) d\vec{y} \\
  & \approx & e^{\Psi(\vec{y}_{\mathrm{Laplace}})} \int e^{-\frac{1}{2} (\vec{y} - \vec{y}_{\mathrm{Laplace}})^T \Sigma_{\mathrm{Laplace}}^{-1} (\vec{y} - \vec{y}_{\mathrm{Laplace}})} d\vec{y} \\
  & = & p(\vec{r}|\vec{y}_{\mathrm{Laplace}}) p(\vec{y}_{\mathrm{Laplace}}) \int e^{-\frac{1}{2} (\vec{y} - \vec{y}_{\mathrm{Laplace}})^T \Sigma_{\mathrm{Laplace}}^{-1} (\vec{y} - \vec{y}_{\mathrm{Laplace}})} d\vec{y}
\end{eqnarray*}
In contrast to a maximum likelihood solution, $p(\vec{r})$ does not just depend on the goodness of fit, as captured by the likelihood $p(\vec{r}|\vec{y}_{\mathrm{Laplace}})$. It also takes into account the prior probability $p(\vec{y}_{\mathrm{Laplace}})$ as well as the uncertainty about the inferred volatility process $\vec{y}$. Thus, the marginal likelihood, automatically
incorporates a trade-off between model fit and model complexity (see chapter 28 of \cite{MacKay:itp} for details). For
this reason, in Bayesian statistics, the marginal likelihood is often
used in model selection.

Here, we use the marginal likelihood to compare different structural
assumptions about the underlying volatility process. In the standard
exponential stochastic volatility model $y_t$ is modelled as an
Ornstein-Uhlenbeck process with covariance function $k^{OU}_{t,t'} =
\frac{\beta^2}{2 \alpha} e^{- \alpha |t - t'|}$. Here,
$\frac{1}{\alpha}$ can be considered as a lengthscale as it gives the
correlation time of the process. The Ornstein-Uhlenbeck process has
continuous, but nowhere differentiable sample paths. In machine
learning rather smooth Gaussian processes with differentiable sample
paths are usually preferred. Table~\ref{tab:GP} gives some examples of
popular Gaussian processes and their covariance functions.

More complex processes, combining the properties of different
processes, can then for example be obtained by adding the corresponding covariance
functions. In all our experiments below, we have added a bias
covariance function to increase the probability to draw sample paths
which are shifted away from zero. This allows our models to easily
represent processes $y_t$ with a non-zero mean. Overall, we fit and
compare processes based on the following covariance structures:
\begin{description}
\item[OU] Bias $+$ Ornstein-Uhlenbeck
\item[RBF] Bias $+$ Squared exponential
\item[RatQuad] Bias $+$ Rational quadratic
\item[RBF\_RBF] Bias $+$ Squared exponential $+$ Squared exponential
\item[OU\_OU] Bias $+$ Ornstein-Uhlenbeck $+$ Ornstein-Uhlenbeck
\end{description}
The model OU is just the standard exponential Ornstein-Uhlenbeck model
\cite{Masoliver2005}, while OU\_OU denotes a two-factor version of
this model incorporating two independent volatility factors with
different time scales. Similarly, the RBF\_RBF model is a two-factor
extension of the model RBF.

To fit the different models we have then optimized the kernel
parameters, e.g. the lengthscale of the Ornstein-Uhlenbeck process, in
order to maximize the marginal likelihood of the observed
returns. Optimizations have been caried out using the Gaussian process
toolbox GPy \cite{gpy2014} with gradient descent using the
Broyden--Fletcher--Goldfarb--Shanno (BFGS) algorithm from 100
different initial conditions.  Using a range of different stock market
data (all downloaded from Yahoo finance) we obtain results as shown in Table~\ref{tab:marginal}.

Our data sets cover very different time periods as well as sectors and
no consistently best model can be indentified.  Nevertheless, we see a
preference for models including two different time scales as either
RBF\_RBF or OU\_OU performs best, closely followed by the RatQuad
kernel. The later model exhibits a power-law decay of volatility
correlations which is considered as a stylized fact of stock market
data. It remains to be investigated why we cannot conform this
finding. Nevertheless, also the best fitting two-factor models have a
long time scale implying that volatility is correlated over several
weeks. In this regard our findings agree with previous results of
econometric studies which strongly suggest the use of two-factor
models which allow ``breaking the link between tail thickness and
volatility persistence'' \cite{Tauchen2003} as needed to fit empirical
data.

As an example, \fig{FitVol} shows the hidden volatility implied by the
Ornstein-Uhlenbeck based (OU) as well as the best model (RBF\_RBF) on the S\&P 500 data. The shaded region shows
the $95$\% region ($\pm 2$ standard deviations from the mean) of the posterior and demonstrates the large
uncertainty that remains about the volatility. When predicting the
volatility for the next 100 days (red curve), the uncertainty is even
more pronounced and increases quickly towards the a-priori uncertainty of the stationary volatility distribution. Note that this uncertainty is not due to parameter uncertainty which have been fixed after fitting. The situation is the same as in the Heston model were we also assumed the parameters known and the high uncertainty is due to the stochasticity of the underlying volatility process. It thus appears to be an intrinsic property of stochastic volatility models, especially if the volatility process varies quickly as in the case of realistic parameters.

\begin{figure}[h]
  \includegraphics[width=.8\linewidth]{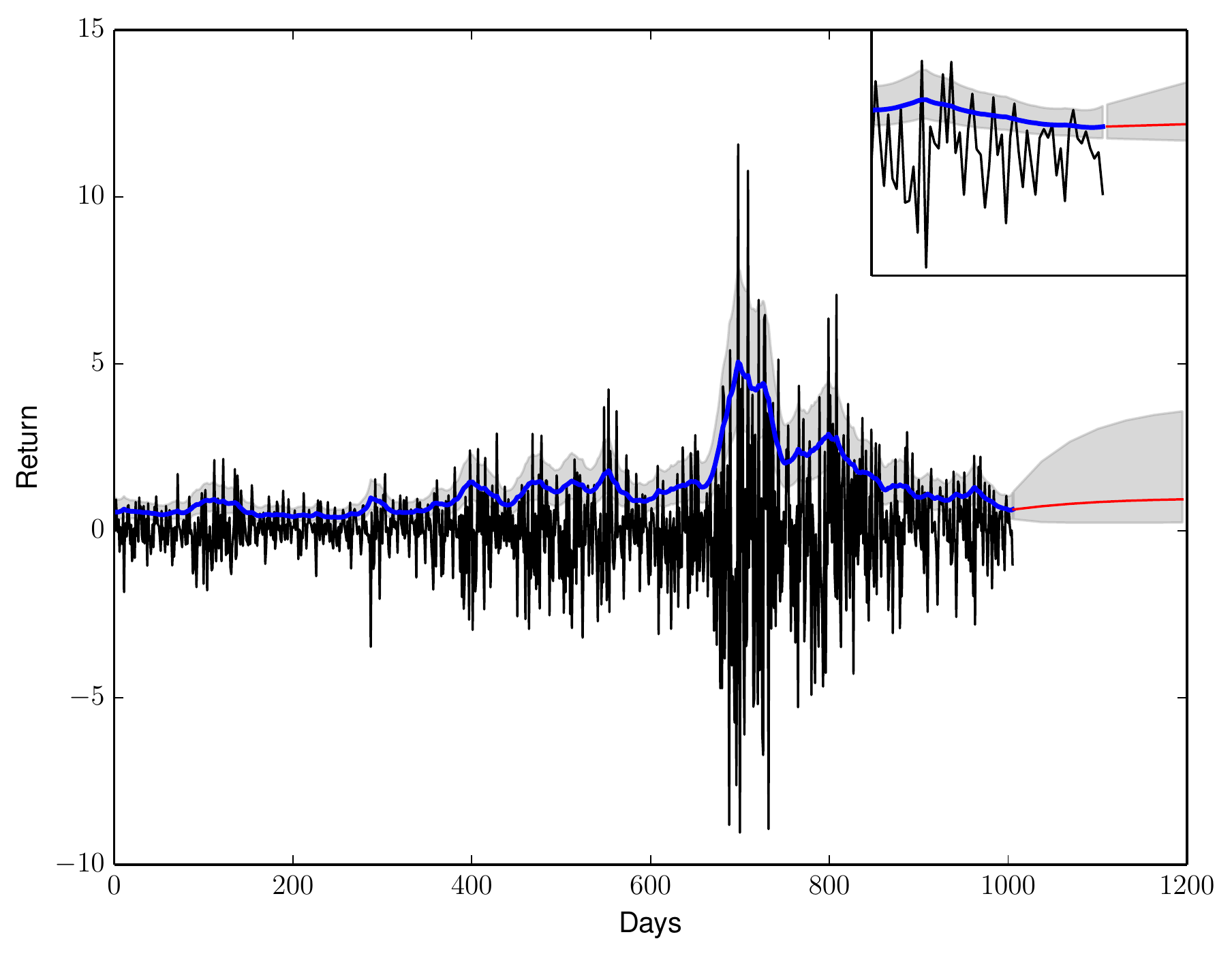} \\
  \includegraphics[width=.8\linewidth]{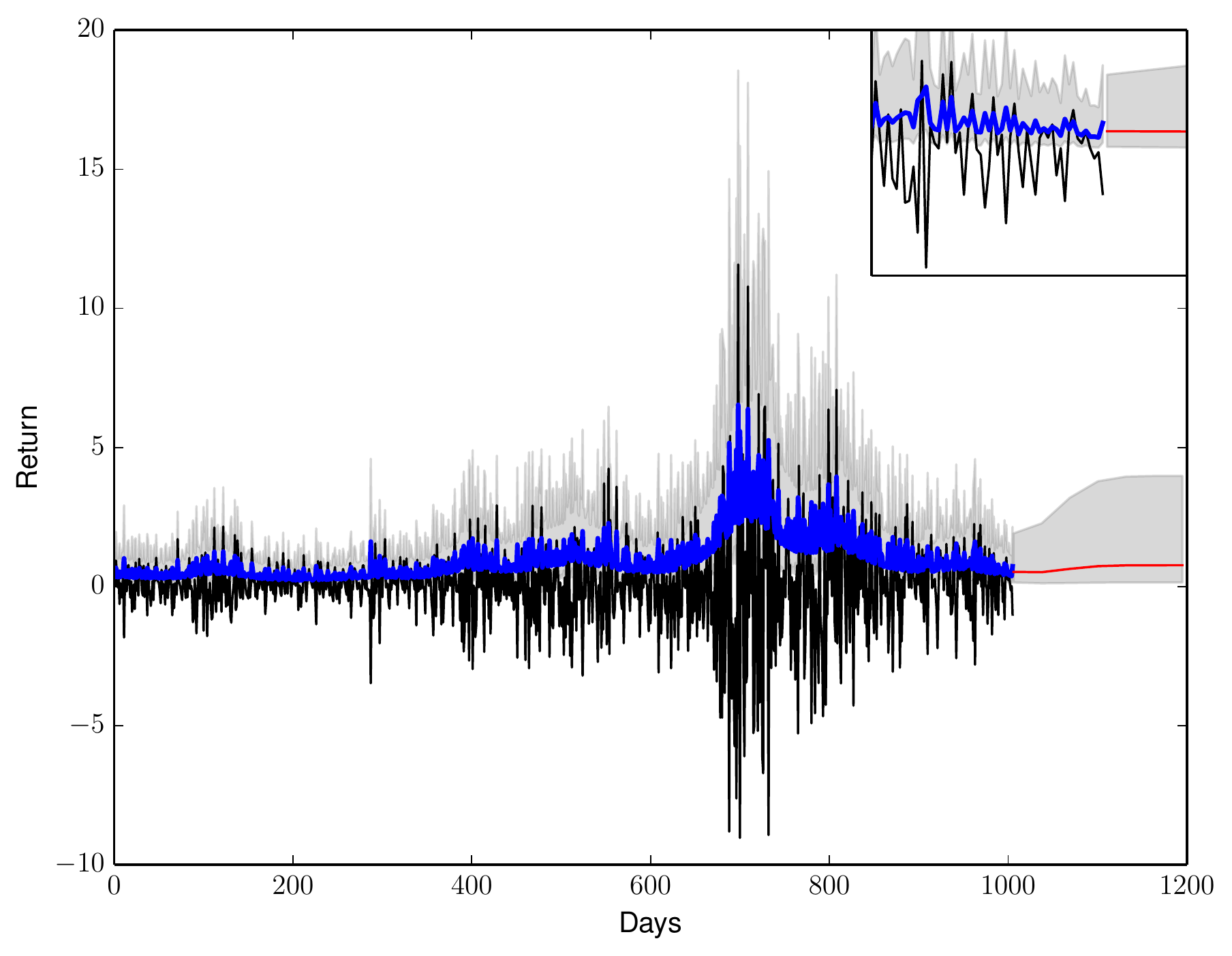}
  \caption{\label{fig:FitVol} Exponential stochastic volatility model
    with OU (top) and RBF\_RBF (bottom) covariance function fitted to
    the S\&P 500 data. The daily returns are shown in black with the
    infered volatility in blue. The grey region shows the large
    uncertainty ($95$\% credibility region) which increases quickly
    when predicting into the future (red curve). The inset shows a
    zoom on the last 50 data points.}
\end{figure}

\subsubsection{Information gain}
As before, we can use information theory to quantify the amount of
information obtained about the volatility. While the mutual
information quantifies the average uncertainty, here we are interested
in the information about $\vec{\sg}$ that was obtained from the
actually observed prices. In \cite{DeWeese1999} it is suggested to
consider the difference in entropy between the prior $p(\vec{\sg})$
and the posterior $p(\vec{\sg}|\vec{S})$ as the information gain from
the actual prices $\vec{S}$. Thus, instead of averaging over all
possible price realizations as in the conditional entropy
$h(\vec{\sg}|\vec{S})$, one computes the entropy of the posterior
distribution conditioned on the actually observed prices, i.e. $- \int
p(\vec{\sg}|\vec{S}) \log p(\vec{\sg}|\vec{S}) d\vec{\sg}$.
In our case, these entropies can be evaluated since the posterior
distribution $p(\vec{y} | \vec{S})$ is approximated by a Gaussian and
the prior $p(\vec{y})$ was assumed to be a Gaussian in the first
place. For a $d$-dimensional multivariate Gaussian with covariance
matrix $\Sigma$ the differential entropy is given by $\frac{1}{2} \log
(2 \pi e)^d |\det \Sigma|$. To have a fair comparison with the results
obtained for the Heston model, we have taken care that the uncertainty
about the mean level of the volatility process did not contribute to
the information gain. Thus, as in the Heston model, we assumed that
the parameter specifying the stationary mean of the volatility process
is known. Formally, we accomplish this by removing the bias component
of the covariance function before computing the information gain. Note
that the entropy of a Gaussian distribution does not depend on the
mean of the distribution, and thus adjusting the mean of the prior
distribution to its stationary value is not necessary when computing
the information gain.

\begin{figure}[h]
  \begin{tabular}{cc}
    OU model & RBF\_RBF model \\
    \includegraphics[width=0.45\linewidth]{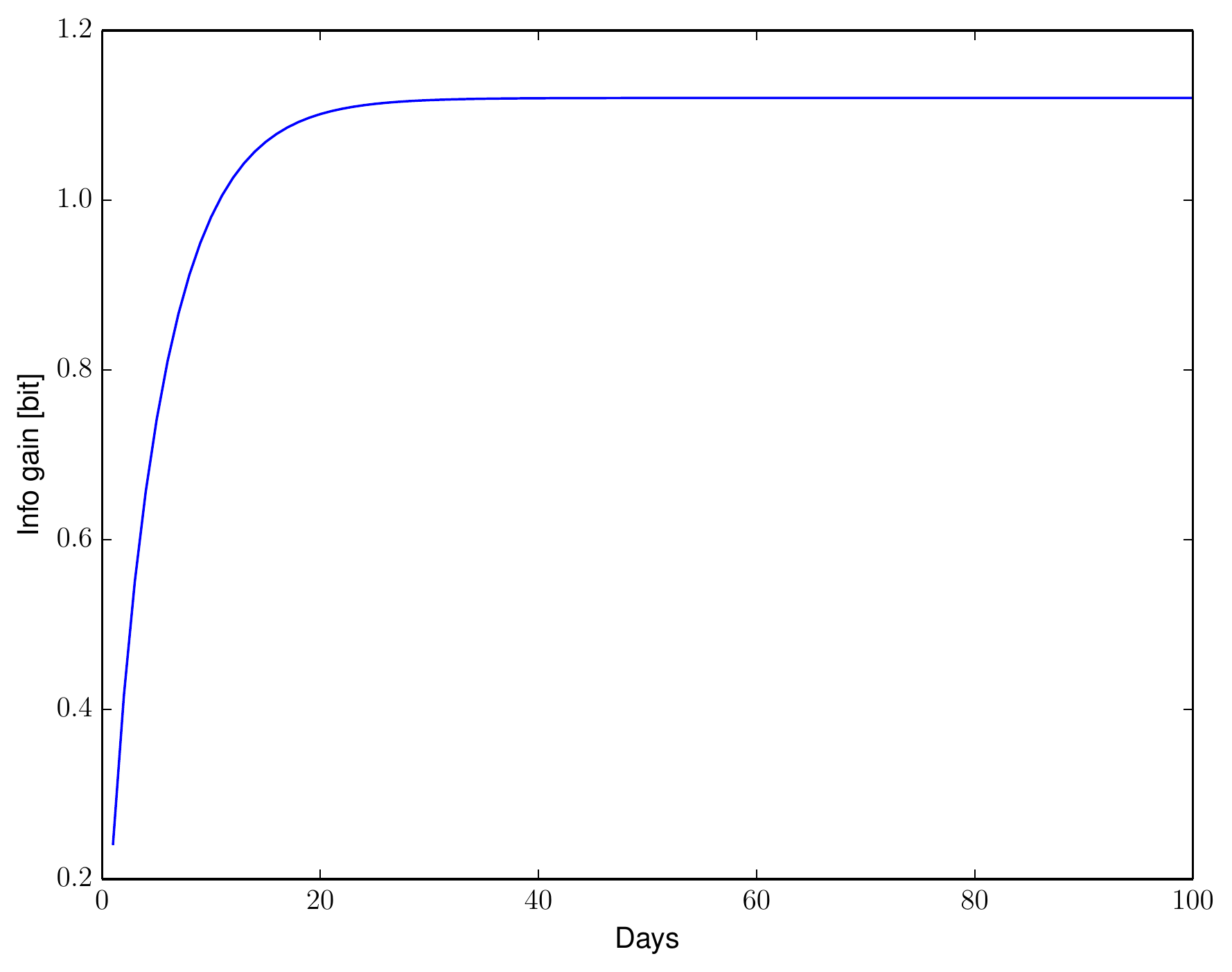} &
    \includegraphics[width=0.45\linewidth]{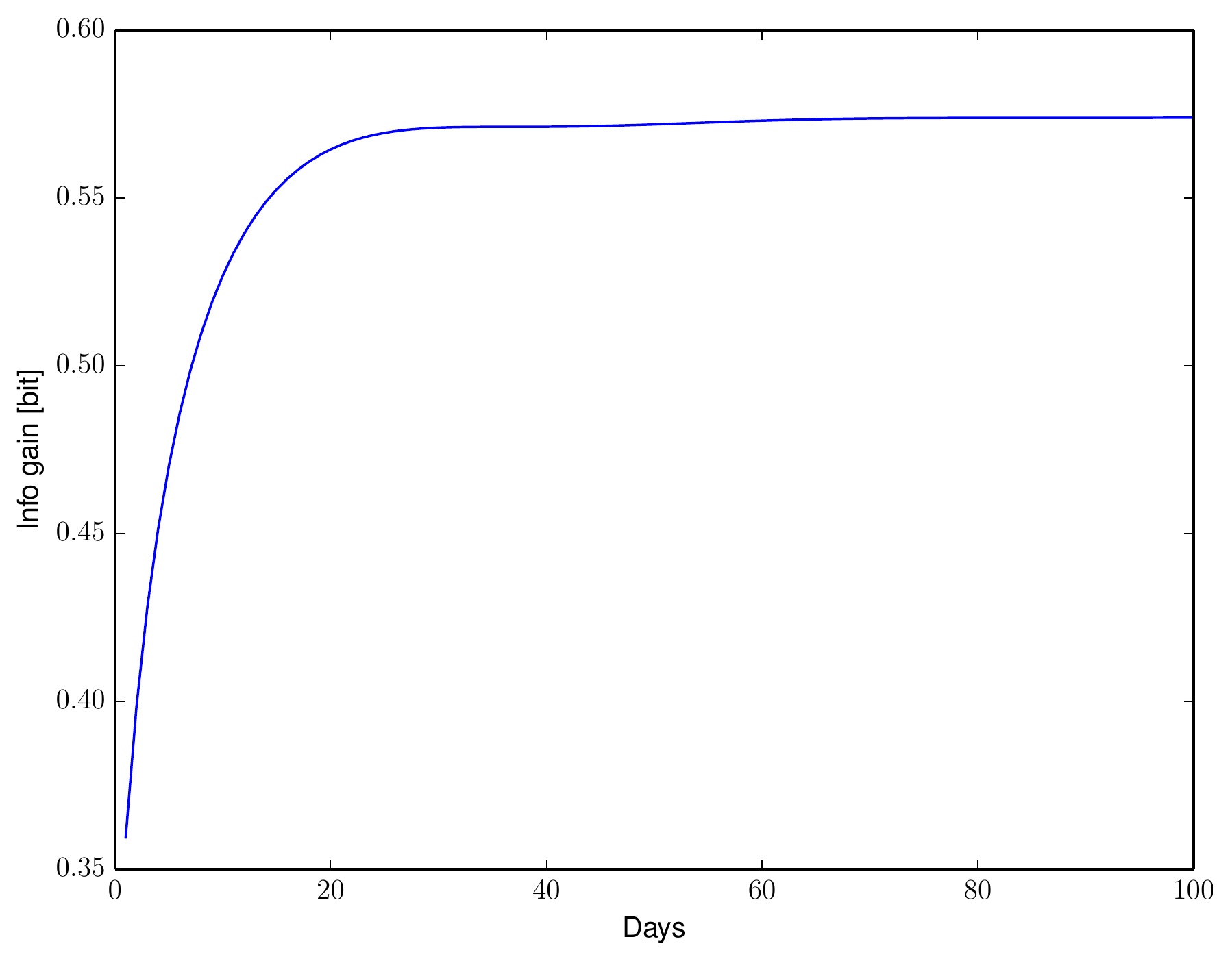}
  \end{tabular}
  \caption{\label{fig:InfoGain} Information gain about the
    instantaneous volatility $\sg_t$ ($t = 1006$) based on
    successively longer sequences of preceding daily returns
    $r_{t-\tau+1}, \ldots, r_t$. In both models the information gain
    saturates after about 20 days.}
\end{figure}
In the example shown in \fig{FitVol} fitted on $4$ years of data
($1006$ trading days) in total $64.5$ (OU) and $251$ bits (RBF\_RBF)
are obtained about the volatility $\vec{\sg}$.  The large difference
reflects the fact that the OU model is less flexible and thus its
a-priori uncertainty is already much smaller. Further, the results are
consistent with the low information gain found for the Heston model as
the above numbers correspond to just $0.06$ (OU) respectively $0.25$
(RBF\_RBF) bits per observation.  While this is comparable to the low
values in the Heston model, especially when assuming $\rho = 0$, the
information about each single observation can well be larger due to
the autocorrelation of the volatility process. To illustrate this
effect, \fig{InfoGain} shows the information gained about the
instantaneous volatility $\sg_{1006}$, corresponding to the last
trading day in our data set, when observing the daily returns
preceding it. One clearly sees that observing successive returns does
reduce the uncertainty about the volatility, but after just about 20
days the information gain saturates. Thus, even observing years of
data would not improve volatility estimates and just about 1 bit of
information can be obtained. Again, we see that inferring volatility
is notoriously hard in stochastic volatility models and volatility
forecasts will be rather imprecise for quite fundamental information
theoretic reasons.


\begin{center}
\begin{table}[h]
  \begin{tabular}{l|ll}
    Name & Covariance & Properties \\ \hline
    Bias & 1 & constant sample paths \\
    Squared Exponential & $e^{- \frac{(t-t')^2}{l^2}}$ & infinitely differentiable sample paths\\
    Rational Quadratic & $\left( 1 + \frac{(t - t')^2}{2 \alpha l^2} \right)^{-\alpha}$ & infinite mixture of lengthscales
  \end{tabular}
  \caption{\label{tab:GP} Gaussian processes that are popular in machine learning.}
\end{table}

\begin{table}[h]
  \begin{tabular}{llll}
    Description & Date & Model & log. marginal likelihood \\ \hline
    Apple Inc. & {\numdate \daterange{2000-01-01}{2003-01-01}} & RBF\_RBF & 1423.22 \\
    & & RatQuad & 1421.33 \\
    & & OU\_OU & 1412.74 \\
    & & OU & 1412.23 \\
    & & RBF & 1403.32 \\
    DAX & {\numdate \daterange{1992-01-01}{1995-01-01}} & OU\_OU & 2509.23 \\
    & & RBF\_RBF & 2505.98 \\
    & & RatQuad & 2498.64 \\
    & & RBF & 2493.47 \\
    & & OU & 2486.51 \\
    Exxon Mobile & {\numdate \daterange{1986-01-01}{1990-01-01}} & RBF\_RBF & 2923.20 \\
    & & OU\_OU & 2921.43 \\
    & & RatQuad & 2906.43 \\
    & & OU & 2897.08 \\
    & & RBF & 2881.11 \\
    S\&P 500 & {\numdate \daterange{2006-01-01}{2010-01-01}} & RBF\_RBF & 3129.78 \\
    & & OU\_OU & 3129.58 \\
    & & RatQuad & 3093.42 \\
    & & OU & 3083.78 \\
    & & RBF & 3072.75
  \end{tabular}
  \caption{\label{tab:marginal} Results of model comparison on
    different stock market data sets.}
\end{table}
\end{center}

\section{Conclusions}

In this paper, we have analysed stochastic volatility models from an
information theoretic perspective. In particular, we have asked how
much information about the hidden volatility can be obtained from
observing stock prices. First of all, the numerical results are rather surprising, at least for the
authors. We have had some doubts on the information content of stocks
about the volatility which is driving them in the Heston model and
related stochastic volatility models, but we did not expect such a
small residual value of at most $0.5$ bit for realistic parameters
taken from the S{\&}P 500.  Second, mutual information
$I(\sg_t, S_t)$ seems to vanish in the long run or becomes at least
meagre. Hence, the processes appear to be almost independent even
though the adjusted log-return process $t \ra x_t$ is entirely driven
by its variance $v_t$, see \eq{log_return}. Third, it is not only difficult
to infer the instantaneous volatility from stock data, it is mostly useless: \fig{info} and \fig{ratio}
imply that knowledge of the instantaneous volatility $\sg_{n\tau}$ just slightly improves
the prediction of $S_{(n+1)\tau}$ given $S_{n\tau}$. Regarding our results in
the third section, the stock process $S_n \ra S_{(n+1)\tau}$ is almost a Markovian
process in its own right. Fourth, the strong dependence of the mutual 
information on the instantaneous correlation $\rho$ is also surprising: 
again the stock and volatility processes are nearly independent if there is 
no correlation. \\
We find a similar behaviour in related stochastic volatility models
which we fitted to actual stock price data. Again, we observe a
large uncertainty when volatility is infered from market prices. This
has severe implications for volatility forecasting and sheds doubt on
the standard practice of comparing models based on their forecasting
error \cite{Hansen2005}. In this case, due to uncertainty about the model parameters,
forecasting should be even less precise than within a model that is
perfectly specified.  A possible remedy could be to incorporate option
prices as well. Volatility is commonly inferred from options not from
stocks.  If Heston's model holds true, what is the information between
the price of an option and the volatility of the underlying
stock. Furthermore, how does this value change if one increases the
number of options considered? While we do not yet know the answer to
these questions, it might well be that substantial uncertanties remain
and volatility estimates, at least from stochastic volatility models,
should be considered as unreliable for quite fundamental information
theoretic reasons.

\section{Funding}

This work was supported by the
European Research Council under the European Union's Seventh Framework
Programme (FP7/2007-2013)/ERC grant agreement no. 318723.  Nils
Bertschinger thanks Dr. h.c. Maucher for funding his position.

\bibliographystyle{abbrv}
\bibliography{Heston}


\end{document}